\documentclass[11pt,reqno]{amsart}

\usepackage{amsmath, amsthm, amscd, amsfonts, amssymb, graphicx, color, mathrsfs, extarrows}
\usepackage{mathscinet}
\usepackage[a4paper,left=3cm,right=3cm,top=2cm,bottom=4cm,bindingoffset=5mm]{geometry}
\usepackage{fancyhdr}
\usepackage{comment}
\usepackage{float}
\usepackage{diagbox}
\usepackage{eurosym}
\usepackage[utf8]{inputenc}
\usepackage{amsmath}
\usepackage{amsfonts}
\usepackage{version}
\usepackage{mathtools}
\usepackage{amscd}
\usepackage{amssymb}
\usepackage{dsfont}
\usepackage{amsthm}
\usepackage{thmtools}
\usepackage{graphicx}
\usepackage{mathrsfs}
\usepackage{todonotes}
\usepackage{xcolor}
\usepackage{todonotes}
\usepackage{setspace}
\usepackage{enumerate}
\usepackage[english]{babel}
\usepackage{xcolor}
\usepackage[colorlinks=true,linkcolor=blue]{hyperref}
\usepackage[all]{hypcap}
\usepackage{newtxmath}
\usepackage{pgfplots}
\usepackage{caption}
\pgfplotsset{compat=1.18}
 
\allowdisplaybreaks

\newtheorem{theorem}{Theorem}
\newtheorem{proposition}[theorem]{Proposition}
\newtheorem{corollary}[theorem]{Corollary}
\newtheorem{lemma}[theorem]{Lemma}

\theoremstyle{definition}

\newtheorem{remark}[theorem]{Remark}
\newtheorem{example}[theorem]{Example}
\newtheorem{definition}[theorem]{Definition}

\newcommand{\cB}{\mathcal{B}}
\newcommand{\cC}{\mathcal{C}}

\newcommand{\cF}{\mathcal{F}}

\newcommand{\cM}{\mathcal{M}}

\newcommand{\cP}{\mathcal{P}}

\newcommand{\cR}{\mathcal{R}}

\newcommand{\cW}{\mathcal{W}}
\newcommand{\cX}{\mathcal{X}}

\newcommand{\E}{\mathbb{E}}

\newcommand{\N}{\mathbb{N}}
\renewcommand{\P}{\mathbb{P}}

\newcommand{\R}{\mathbb{R}}

\newcommand{\al}{\alpha}

\newcommand{\de}{\delta}
\newcommand{\ep}{\varepsilon}

\newcommand{\ga}{\gamma}
\newcommand{\om}{\omega}

\newcommand{\Om}{\Omega}

\newcommand{\la}{\lambda}

\renewcommand{\d}{{\rm d}}
\newcommand{\cmp}{{\rm c}}
\newcommand{\eins}{\mathds{1}}

\newcommand{\deq}{\stackrel{\rm d}{=}}

\DeclareMathOperator{\VaR}{\rm VaR}
\DeclareMathOperator{\TVaR}{\rm TVaR}

\DeclareMathOperator{\esssup}{ess\,sup}

\DeclareMathOperator{\spear}{sp}

\DeclareMathOperator{\var}{var}
\DeclareMathOperator{\cor}{cor}
\DeclareMathOperator{\cpl}{cpl}

\DeclareMathOperator{\pr}{pr}

\title[Hidden Dependence and Aggregate Tail Risk]
{Hidden Dependence and Aggregate Tail Risk}

\author{Corrado De Vecchi}
\address{Department of Economics, University of Verona, Italy}
\email{corrado.devecchi@univr.it}

\author{Max Nendel}
\address{Department of Statistics and Actuarial Science, University of Waterloo, Canada}
\email{mnendel@uwaterloo.ca}

\author{Steven Vanduffel}
\address{Solvay Business School, Vrije Universiteit Brussel, Belgium}
\email{steven.vanduffel@vub.be}

\thanks{The first and second named authors are grateful for funding from the Natural Sciences and Engineering Research Council of
Canada via Discovery Grant no.\ RGPIN-2025-04219.\ The last-named author acknowledges funding from the Research Foundation -- Flanders (FWO) through Research Grant No.\ G093024N}

\date{\today}

\begin{document}

\begin{abstract}
We study risk aggregation problems for arbitrary non-decreasing aggregation functions and tail risk measures under dependence uncertainty in a distributionally robust setting.\ To this end, we introduce the notion of hidden dependence for random vectors, which is built on the concepts of risk concentration and common tail events developed in Wang and Zitikis (2020).\
 We show that, starting from a tail event $A$ of the aggregate loss for an arbitrary random vector $Y$, one can construct a random vector with hidden dependence that dominates $Y$ on the tail event $A$.\ We then focus on the case in which model uncertainty is described by small perturbations of the distribution of a random vector with respect to a suitable probability distance without changing the marginals.\ We show that these perturbations of the reference distribution are compatible with hidden dependence and thus lead to the same worst-case risk bounds as in the unconstrained case for arbitrary $\ga$-tail risk measures with a suitable level $\ga$.\ Finally, we apply our results in a credit risk context and quantify the potential underestimation of portfolio risk arising from uncertainty in the dependence structure. In particular, we show that even small deviations from a reference Gaussian dependence model can, in principle, justify dramatic increases in capital requirements.


\end{abstract}

\maketitle

\section{Introduction}

Regulators require financial institutions and insurers to hold sufficient capital to protect depositors and policyholders and to support financial stability.\ These capital requirements are typically based on estimates of aggregate risk, obtained by combining the risks of individual positions, business lines, or entities.\ Their reliability therefore depends not only on the marginal behaviour of these risks, but also on the way their dependence is modelled.\ This is particularly important for risk measures, such as Value-at-Risk (VaR) and Tail Value-at-Risk (TVaR), which play a central role in capital requirements and solvency assessments.\ Since these measures focus on rare losses, they can be highly sensitive to modelling assumptions, especially those concerning dependence, see, for example, \cite{BRV2022, Embrechts2002}.\ This concern is explicitly recognized by regulators and actuarial bodies. For instance, the Actuarial Association of Europe (2017) emphasizes that ``model risk cannot be disregarded,'' since many models may be consistent with the available data and the final model choice is therefore partly subjective, cf.\ \cite{AAE2017}. Similarly, the Basel Committee on Banking Supervision (2019) encourages banks to provide evidence on the uncertainty surrounding capital requirement models by identifying key assumptions and estimating uncertainty bounds, cf. \cite{BCBS2019MarketRisk}.

Although institutions can often estimate the marginal distributions of individual losses with reasonable accuracy, characterizing the dependence structure between these losses is substantially more difficult.\ A standard approach to risk aggregation relies on correlation-based methods, which often implicitly assume Gaussian dependence structures.\ These methods are appealing because of their analytical tractability and simplicity, but they are unable to capture more complex forms of dependence, especially during extreme events, potentially leading to a severe underestimation of tail risk, cf.\ \cite{Embrechts2002}.\ Copula-based approaches offer considerably greater flexibility and can model dependence patterns that better reflect empirical observations, see, e.g., \cite{Czado2010}.\ However, choosing and calibrating an appropriate copula remains challenging, and reliable estimation typically requires rich datasets that are rarely available in financial and insurance applications.\ For instance, assessing the risk of a loan portfolio would require data containing observations of multiple simultaneous defaults, yet such events occur infrequently.\ As a result, although substantial information on individual risks may be available, uncertainty about their interdependence usually leads to considerable uncertainty in estimates of the aggregate risk, cf.\ \cite{BRV2022}.\  
These difficulties and their importance for society have motivated a growing literature on risk aggregation under dependence uncertainty.

The classical setting assumes fixed marginal distributions while leaving the dependence structure completely unspecified.\ As a consequence, the resulting model risk --- often quantified as the difference between the worst-case value and the value obtained under a benchmark dependence model for risk measures such as VaR or TVaR --- can be substantial, see, e.g., \cite{EmbrechtsWangWang2015, PuccettiRueschendorf2012b}.\ As an  illustration, under dependence uncertainty, the worst-case VaR of an aggregate portfolio is at least as large as the sum of the individual VaRs and may approach the sum of the individual TVaR values.\ Consequently, in the absence of dependence information, estimates of capital requirements may fail to reflect the actual risk and may not provide a sufficiently informative assessment of aggregate risk.

More recent contributions therefore focus on incorporating partial dependence information, for example, through copula constraints, factor structures, variance information, distance or correlation constraints, see, among others, \cite{BernardVanduffel2015ModelRisk, Bignozzietal2015, PuccettiRueschendorfManko2016, BernardRueschendorfVanduffelWang2017, LiShaoWangYang2018, LuxPapapantoleon2019, fontana2021model, bernard2023impact}.\ This paper contributes to this literature by studying worst-case tail risks under uncertain dependence.\ Our starting point is the notion of a tail risk measure introduced in \cite{liu2021theory}, which formalizes risk measures that depend only on the part of a distribution beyond a specified quantile threshold.\ This class comprises prominent regulatory risk measures such as VaR and TVaR, and provides a natural framework for aggregation problems in which only extreme losses are of interest.\ We consider arbitrary non-decreasing aggregation functions and characterize worst-case values of such tail risk measures when marginal distributions are fixed but the dependence structure remains uncertain.

Central to this paper is the notion of \emph{hidden dependence}, which we introduce.\ Hidden dependence describes a situation in which a given copula properly captures the dependence among the random variables in a random vector, provided that their largest realizations are not observed.\ However, hidden dependence provides no information about the dependence structure, conditional on an adverse scenario described by a common tail event, that is, when  large realizations are observed simultaneously with probability $1-\gamma$ for $\gamma \in (0,1)$. 
Our first contribution is to show how, starting from a tail event of the aggregate loss for an arbitrary random vector, one can construct another random vector with hidden dependence that almost surely dominates it on that tail event. As a consequence, for the purpose of maximizing a $\gamma$-tail risk measure, the search over all admissible dependence structures can be reduced to dependence structures exhibiting hidden dependence.

We then study ambiguity sets generated by small perturbations of a reference distribution, measured by a consistent probability distance.\ Our second contribution is to show that such perturbations are compatible with hidden dependence.\ This implies that, even when admissible joint distributions are required to remain close to a reference model, the resulting worst-case tail risk bounds may coincide with those obtained under full dependence uncertainty.\ In particular, however small the allowed perturbation around a given reference model may be, there exists a level $\gamma \in (0,1)$ such that the worst-case value of the $\gamma$-tail risk measure is unaffected by the distance constraint. Moreover, through the connection between consistent probability distances and the broad class of regular dependence measures considered in \cite{CorradoMaxJan}, this conclusion extends to ambiguity sets described in terms of dependence measures. Hence, the same worst-case bounds may arise not only under small distributional perturbations of a reference model, but also when partial dependence information is imposed through a broad class of dependence constraints.

This shows that dependence information does not necessarily improve risk bounds when tail risk measures are evaluated at high confidence levels. In other words, even if the multivariate distribution is known almost perfectly, this information may become irrelevant in the far tail:\ the worst-case value may depend only on the marginal distributions and thus coincides with the unconstrained bound.\ In such cases, existing results and techniques from the literature can be applied directly to compute the unconstrained bounds.\ For instance, for a coherent tail risk measure, the upper bound is attained by a copula exhibiting co-monotonicity in the tail. One may argue that such extreme dependence structures are unrealistic, and that this result merely reflects an insufficient specification of dependence information. However, in the absence of sufficient tail data, it is appropriate to leave open a broad range of possible tail behaviours. Moreover, empirical evidence suggests that extreme co-movement among risks does occur; for various illustrations we refer, among others, to \cite{DasUppal2004, ContKokholm2013, ContWagalath2016}.

A practical implementation of our second main result requires understanding the relationship between the radius $\ep$ of the ambiguity ball around the reference distribution and the tail level $\gamma$. Our third contribution is to establish, for $p$-Wasserstein distances with $p\in[1,\infty)$, how to compute, for a given tail level $\gamma$, a minimum value for the radius of the ambiguity ball at which the resulting constrained bound coincides with the unconstrained bound, and vice versa.\ As an application, we show that, even when default probabilities are known and only small perturbations around a reference Gaussian copula are acceptable, model risk can be devastating:\ the worst-case VaR may be as large as the total exposure.

A framework similar to ours is that of \cite{EKP2020}, where also fixed marginal distributions and ambiguity at the level of dependence is considered.\ However, their work focuses on the numerical computation of worst-case expectations for aggregate losses, while we provide structural  conditions, based on hidden dependence, for worst-case tail risk measures.

Our results are also closely related to \cite{CorradoMaxJan}, who show that an arbitrarily small amount of positive dependence, measured through a broad class of dependence measures, can be explained by independence of the risks unless a tail event occurs; in that case, the risks become co-monotonic, leading to tail risk levels that match those under co-monotonicity. We show that adding dependence information, expressed via dependence constraints or through distributional perturbations around a reference model, remains insufficient to control aggregate tail risk and that the worst-case risk is attainable.

Overall, the present paper shows that tail risk can be highly sensitive to forms of dependence that are difficult to detect from the body of the distribution or from global dependence summaries.\ This reinforces the need for caution when interpreting risk estimates based on fitted multivariate models, especially when the quantity of interest is a high-level tail functional.\ More broadly, our findings demonstrate that seemingly mild uncertainty concerning dependence structures may suffice to recover worst-case tail scenarios, highlighting the persistent vulnerability of tail risk assessment to hidden forms of dependence.

The rest of the paper is organized as follows.\ Section 2 introduces the notation and recalls the notion of a tail risk measure.\ Section 3 develops the main theoretical framework.\ We introduce hidden dependence, show how it characterizes worst-case tail risk under dependence uncertainty, and prove that distance constraints around a reference dependence model may fail to reduce the worst-case bound. We also discuss how these results apply when dependence information is expressed through dependence measures rather than probability distances.\ Section 4 considers risk measures that are not necessarily tail risk measures, focusing in particular on Value-at-Risk and expectiles. Section 5 provides quantitative implications of the theory.\ We derive explicit bounds for $p$-Wasserstein distances with $p\in[1,\infty)$ and for several prominent choices of dependence measures.\ Finally, in Section 6, we apply our results to credit risk portfolios and demonstrate that even small deviations from a reference Gaussian dependence model can lead to substantial increases in capital requirements.

\section{Setup and Notation}

Throughout, let $n\in \N$ and $(\Om,\cF,\P)$ be an atomless probability space.\ Recall that a probability space is atomless if and only if it supports a uniformly distributed random variable, cf.\ \cite[Proposition A.31]{foellmer2016stochastic}.\ In what follows, $\cP(\R)$ and $\cP(\R^n)$ denote the set of all Borel probability measures on $\R$ and $\R^n$, respectively.

For a random variable $Z\sim \mu\in \cP(\R)$, we write $F_Z:=F_\mu$ for its cumulative distribution function and define
\[
   \VaR^\al (Z):=F_Z^{-1}(\al):= \inf \big\{ a\in \R\,\big|\,\P(Z>a)\leq 1-\al\big\}
\]
as the left-continuous version of the \textit{Value-at-Risk} (VaR) at level $\alpha\in(0,1)$.\ 

Furthermore, the \textit{Tail Value-at-Risk} (TVaR) at level $\alpha\in(0,1)$ with $\E[|Z|]<\infty$ is defined as
\[\text{TVaR}^{\alpha}(Z)=\frac{1}{1-\alpha}\int_\alpha^1F^{-1}_{Z}(u)\d u.\]  

For a random vector $X=(X_1,\ldots, X_n)$, we also use the notation
\[
F_X^{-1}(u):=\big(F_{X_1}^{-1}(u_1),\ldots, F_{X_n}^{-1}(u_n)\big), \quad\text{for }u=(u_1,\ldots, u_n)\in (0,1)^n,
\] 
and $\P_X\in \cP(\R^n)$ for its $n$-variate distribution, given by
\[
\P_X(B):=\P(X\in B)\quad \text{for all }B\in \cB(\R^n),
\]
where $\cB(\R^n)$ denotes the Borel $\sigma$-algebra on $\R^n$.\ As usual, a copula $c\colon [0,1]^n\to [0,1]$ is the cumulative distribution function of a random vector $U=(U_1,\ldots, U_n)$ with $U_i$ uniform on $(0,1)$ for $i=1,\ldots, n$.\

We briefly recall the notion of a tail risk measure, introduced in \cite{liu2021theory}.\ To that end, let $\cX$ be a nonempty set of random variables and $\ga\in (0,1)$.\ For any random variable $Z\in \cX$, let $U_Z$ be a uniformly distributed random variable with $Z=F_Z^{-1}(U_Z)$ $\P$-almost surely.\ The existence of such a random variable follows, for example, from \cite[Lemma A.32]{foellmer2016stochastic}, recalling that $(\Om,\cF,\P)$ is assumed to be atomless.
\[
Z^\ga:=F_Z^{-1}(\ga+(1-\ga)U_Z)\quad\text{for }\ga\in (0,1).
\]
We point out that, a priori, the definition of $Z^\ga$ also depends on the choice of the uniformly distributed random variable $U_Z$. However, passing from $U_Z$ to another uniformly distributed random variable, say $V$, it follows that
\begin{equation}\label{eq.uniqueness.tail.risk}
 Z^\ga\deq F_Z^{-1}(\ga+(1-\ga)V) \quad\text{for all }\ga\in (0,1),
\end{equation}
 i.e., $Z^\ga$ is unique up to equality in distribution.\
Following \cite[Definition 1]{liu2021theory}, for $\ga\in (0,1)$, we say that a map $R\colon \cX\to \R$ is a \textit{$\ga$-tail risk measure} if it is monotone, i.e.,
\[
R(Z_1)\leq R(Z_2)\quad \text{for all } Z_1,Z_2\in \cX\text{ with }Z_1\leq Z_2 \;\P\text{-a.s.},
\]
and $R(Z_1)=R(Z_2)$ for all $Z_1,Z_2\in \cX$ with $Z_1^\ga\deq Z_2^\ga$.\ We refer to \cite{liu2021theory} for a more detailed discussion of tail risk measures.

 For $\mu_1,\ldots,\mu_n\in \cP(\R)$,
$\cpl(\mu_1,\ldots,\mu_n)$ denotes the set of all multi-marginal couplings, i.e., the set of all $\pi\in \cP(\R^n)$ with
$$
\pi\circ\pr_i^{-1}=\mu_i\quad \text{for all }i=1,\ldots, n,
$$
where $\pr_i\colon \R^n\to \R$ refers to the $i$-th coordinate projection
 for $i=1,\ldots, n$.\ For a nonempty set $\cM\subset \cP(\R)$, we define
$$
\cpl^n(\cM):=\bigcup_{\mu_1,\ldots, \mu_n\in \cM} \cpl(\mu_1,\ldots, \mu_n).
$$ 
Throughout this paper, we consider aggregation functions $f\colon \R^n \to \R$, which are assumed to be (Borel) measurable and non-decreasing, i.e.,
\[
   f(x)\leq f(y)\quad \text{for all }x,y\in \R^n\text{ with }x \leq y
\]
and, occasionally, also left-continuous, i.e., $f(x)=\lim_{y\uparrow x}f(y)$ for all $x\in \R^n$.\

\section{Risk aggregation under dependence uncertainty}\label{sec.risk-aggregation}

In this section, we study a general class of risk aggregation problems of the form
  \begin{equation}\label{eq: mainprob}
      \begin{aligned}
\sup \quad & R\Big( f(X)\Big) \\
\textrm{s.t.} \quad & X_i\sim \mu_i, \textbf{ }i=1,\ldots,n,  \\
& \P_X \in \cC,
\end{aligned} 
  \end{equation}
where $R$ is a tail risk measure, $f\colon\R^n\to \R$ is a non-decreasing measurable aggregation function, $\mu_1,\ldots, \mu_n\in \cP(\R)$ are one-dimensional marginals, and $\cC$ is a set of $n$-variate distributions that are deemed feasible and that are compatible with the marginals $\mu_1,\ldots, \mu_n$.

A key point of our analysis is to show that a specific property of the risks' dependence structure, which we call $(c,\ga)$-hidden dependence, plays a central role for the description of the risk bounds in \eqref{eq: mainprob} in terms of a risk aggregation problem with full dependence uncertainty, i.e., without additional constraints in terms of a set $\cC$, which corresponds to the case $\cC=\cP(\R^n)$.

\subsection{Risk Aggregation with Hidden Dependence}\label{sec.comonotonic}

Before introducing the concept of $(c,\ga)$-hidden dependence, we recall the notion of tail events and common tail events introduced in \cite{WangZitikis2020}. 

\begin{definition}
Let $\ga\in (0,1)$.
\begin{enumerate}
\item[a)] An event $A\in \cF$ is called a \textit{$\ga$-tail event} for a random variable $Z$ if $\P(A)=1-\ga$ and
\begin{equation}\label{eq.ga-tailevent}
 Z(\om)\geq \sup Z(A^\cmp)\quad \text{for all }\om \in A.
 \end{equation}
\item[b)] An event $A\in \cF$ is called a \textit{common $\ga$-tail event} for a random vector $X=(X_1,\ldots, X_n)$ if it is a $\ga$-tail event for $X_1,\ldots, X_n$.
 \end{enumerate}
\end{definition}
Intuitively, a $\ga$-tail event for the random loss $Z$ is an event occurring with probability $1-\ga$ in which the loss $Z$ takes its highest values. Considering a financial portfolio with random losses $X_1,\ldots, X_n$, a $\ga$-tail event is an event in which all the losses simultaneously assume their highest values, and again this event has a probability equal to $1-\ga$.\ 
Finally, we recall that \cite{WangZitikis2020} showed that a random variable $Z$ always admits a $\ga$-tail event.

\begin{remark}\label{rem: tailriskandevents}
   Let $\ga\in (0,1)$.\ Then, by {Lemma \ref{lem: tailevent} c)}, for a given random variable $Z$ with a $\ga$-tail event $A\in \mathcal{F}$, the equality $\P(Z\leq a\vert A) =\P(Z^\ga \leq a)$ holds for any $a\in \R$. Hence, any $\ga$-tail risk measure only depends on the distribution of $Z$ conditioned on a $\ga$-tail event.
\end{remark}
We are now ready to state the definition of $(c,\ga)$-hidden dependence for a given random vector.
\begin{definition}\label{def.hiddendep}
Let $c\colon [0,1]^n\to [0,1]$ be a copula and $\ga\in (0,1)$. Then, we say that a random vector $X=(X_1,\ldots, X_n)$ has \textit{$(c,\ga)$-hidden dependence} if there exists a common $\ga$-tail event $A\in \cF$ such that
 \begin{equation}\label{eq.def.gamma-comonotonic}
 \P(X\leq a\,|\, A^\cmp)=c\big(\P(X_1\leq a_1\,|\,A^\cmp),\ldots, \P(X_n\leq a_n\,|\,A^\cmp)\big)
 \end{equation}
 for all $a=(a_1,\ldots,a_n)\in \R^n$.
\end{definition}
By saying that a random vector has a $(c,\ga)$-hidden dependence, we are essentially specifying the dependence structure of $X$ conditionally on $A^\cmp$, i.e., the complement of the common $\ga$-tail event $A$.\ As for an interpretation, the assumption of $(c,\ga)$-hidden dependence states that the copula $c$ properly captures the dependence among the random variables in $X$, assuming that we do not observe their largest realizations. On the other hand, $(c,\ga)$-hidden dependence does not provide information regarding the dependence structure conditional on the adverse scenario described by the common $\ga$-tail event $A.$
The following lemma offers three equivalent formulations of $(c,\ga)$-hidden dependence that shed light on this assumption.

\begin{lemma}\label{lem:char-hidden}
Let $c\colon [0,1]^n\to [0,1]$ be a copula, $\ga\in (0,1)$, and $X=(X_1,\ldots, X_n)$ be a random vector.\ Then, the following statements are equivalent.
\begin{enumerate}
\item[(i)] $X$ has $(c,\ga)$-hidden dependence.
\item[(ii)] There exist $A\in \mathcal F$ with $\P(A)=1-\ga$ and a copula $c^{\rm tail}$ such that
 \begin{equation}\label{eq: rapgen}
 X_i\overset{\d}{=}F_{X_i}^{-1}\Big(\gamma U_i \mathds{1}_{A^\cmp} + \big(\ga + (1-\ga)V_i\big)\mathds{1}_{A}\Big)\quad \text{for }i=1,\ldots, n,
 \end{equation}
  where $(U_1,\ldots ,U_n)\sim c$, conditionally on $A^\cmp$, and $(V_1,\ldots, V_n)\sim c^{\rm tail}$, conditionally on $A$.
 \item[(iii)] There exists a copula $c^{\rm tail}$ such that $c_X$, given by
 \begin{equation}\label{eq.copula.rep}
    c_X(u)=\begin{cases} \ga+ (1-\ga)c^{\rm tail}\Big(\tfrac{u_1-\ga}{1-\ga},\ldots, \tfrac{u_n-\ga}{1-\ga}\Big), & \text{for }u\in (\ga,1)^n,\\
    \ga c\Big(\tfrac{u_1\wedge \ga}\ga, \ldots, \tfrac{u_n\wedge \ga}\ga\Big),& \text{for } u\in (0,1)^n\setminus (\ga,1)^n,\end{cases}
 \end{equation}
 is a copula of $X$.
 \item[(iv)] There exist $A\in \mathcal F$ with $\P(A)=1-\ga$ and a copula $c^{\rm tail}$ such that \eqref{eq: rapgen} holds
  for all $(U_1,\ldots ,U_n)\sim c$, $(V_1,\ldots, V_n)\sim c^{\rm tail}$ independent of $A$.
 \end{enumerate}
\end{lemma}

\begin{proof}
 We start with the implication (i)$\,\Rightarrow\,$(ii). To that end, assume that $X$ has $(c,\ga)$-hidden dependence.\ Then, by \eqref{eq.cond.var} below, there exist random variables $(U_1,\ldots, U_n)$ with $(U_1,\ldots, U_n)\sim c$, conditionally on $A^\cmp$, and
 \[
 X_i= F_{X_i}^{-1}(\ga U_i)\quad \P\text{-a.s., conditionally on }A^\cmp,\text{ for } i=1,\ldots, n.
 \]
 Moreover, by \eqref{eq.cond.var} below, there exist random variables $V_1,\ldots, V_n$ which are uniform on $(0,1)$, conditionally on $A$, such that
 \[
 X_i= F_{X_i}^{-1}\big(\ga +(1-\ga)V_i\big)\quad \P\text{-a.s., conditionally on }A,\text{ for }i=1,\ldots, n.
 \]
 This proves \eqref{eq: rapgen}.\ Now, the implication (ii)$\,\Rightarrow\,$(iii) is a direct consequence of \eqref{eq: rapgen} together with \eqref{eq.cond.distribution} and the implication (iii)$\,\Rightarrow\,$(iv) is trivial.\ The remaining implication (iv)$\,\Rightarrow\,$(i) follows from the equality
 \begin{align*}
       \P\big(X \leq a \vert A^\cmp\big)&= \P\Bigg( \bigcap_{i=1}^{n} \Big\{F^{-1}_{X_i}\big( \ga U_i^c\big)\leq a_i\Big\} \,\Bigg \vert \, A^\cmp \Bigg) = \P\Bigg( \bigcap_{i=1}^{n} \Bigg\{   U_i^c \leq \frac{F_{X_i}(a_i)\wedge \ga}{\ga}\Bigg\}\, \Bigg \vert \,A^\cmp \Bigg)  \\ 
       &= c\Bigg ( \frac{F_{X_1}(a_1)\wedge \ga}{\ga} ,\ldots, \frac{F_{X_n}(a_n)\wedge \ga}{\ga} \Bigg)\\
   &= c\big(\P(X_1\leq a_1\,|\, A^\cmp),\ldots, \P(X_n\leq a_n\,|\, A^\cmp)\big)
   \end{align*} 
   for all $a\in \R^n$, where, in the last step, we used \eqref{eq.cond.distribution}.
\end{proof}

The following remark shows that common $\ga$-tail events for a random vector $X$ are $\ga$-tail events for the aggregate position $f(X)$ for any non-decreasing aggregation function $f\colon \R^n\to \R$.

\begin{remark}\label{rem.tailevent.aggregate}
 Let $X=(X_1,\ldots, X_n)$ be a random vector, $\ga\in (0,1)$, $A\in \cF$ be a common $\ga$-tail event for $X$, and $f\colon \R^n\to \R$ be non-decreasing and measurable.\ Then, $A$ is a $\ga$-tail event for $f(X)$.\ Indeed, for all $\om \in A$,
 \[
 f\big(X(\om)\big)\geq f\big(\sup X_1(A^\cmp),\ldots, \sup X_n(A^\cmp)\big)\geq \sup [f(X)](A^\cmp).
 \]
\end{remark}

The following proposition forms the backbone of the theoretical results presented in this paper.\ In terms of random losses, it states that one can start with a random vector $Y=(Y_1,\ldots, Y_n)$, fix $\gamma \in (0,1)$ and a copula $c\colon [0,1]^n\to[0,1]$, and construct a random vector $X=(X_1,\ldots, X_n)$ with the same marginal distributions as $Y$ and $(c,\ga)$-hidden dependence such that each component of $X$ dominates the corresponding component of $Y$ almost surely on an event $A\in \cF$, which can be freely chosen among those with probability $1-\ga$.

\begin{proposition}\label{th: cgammadominance}
    Let $c\colon [0,1]^n\to [0,1]$ be a copula, $\ga\in (0,1)$, $Y=(Y_1,\ldots, Y_n)$ be a random vector, and $A\in \cF$ with $\P(A)=1-\ga$.\ Then, there exists a random vector $X=(X_1,\ldots, X_n)$ with $(c,\ga)$-hidden dependence, $X_i\overset{\d}{=}Y_i$ for $i=1,\ldots, n$, common $\ga$-tail event $A$, and
    \[
     X_i\geq Y_i\quad  \P\text{-a.s.\ on }A \quad \text{for }i=1,\ldots, n.
    \]
\end{proposition}

\begin{proof}
 Let $U_1,\ldots, U_n$ be uniformly distributed random variables such that $Y_i=F_{Y_i}^{-1}(U_i)$ for $i=1,\ldots, n$.\ We denote by $F_{\mu_i}$ the conditional distribution function of $U_i$ given $A$, i.e., $$F_{\mu_i}(u):=\P(U_i\leq u\,|\, A)\quad \text{for }u\in (0,1)\text{ and }i=1,\ldots, n.$$ Then, $F_{\mu_i}(U_i)$ is uniformly distributed conditionally on $A$. Moreover let $U_1^c, \ldots, U_n^c$ be uniformly distributed with copula $c$ conditionally on $A^\cmp$. Let $$V_i:=\ga U_i^c\eins_{A^\cmp}+\big(\ga+(1-\ga)F_{\mu_i}(U_i)\big)\eins_A\quad \text{for }i=1,\ldots, n.$$
 Then, for $u\in (0,1)$ and $i=1,\ldots, n$,
 \[
 \ga+(1-\ga)F_{\mu_i}(u)=\ga+\P\big(\{U_i\leq u\}\cap A\big)\geq \ga +(1-\ga)-(1-u)=u.
 \]
 Therefore $V_i\geq U_i$ on $A$ for $i=1,\ldots, n$.\ Now, we define $X_i:=F_{Y_i}^{-1}(V_i)$ for $i=1,\ldots, n$. Then, $X_i\geq Y_i$ $\P$-a.s.\ on $A$ and
 \[
 \P\big(X_i\leq a\,|\, A^\cmp\big)=\P\big(F^{-1}_{Y_i}( \ga U_i^c)\leq a\,\big|\, A^\cmp\big)=\P\bigg(  U_i^c\leq \frac{F_{Y_i}(a)\wedge \ga}{\ga}\,\bigg|\, A^\cmp\bigg)=\frac{F_{Y_i}(a)\wedge\ga }{\ga}
 \]
 for $a\in \R$ and $i=1,\ldots, n$. 
 Moreover, for $a=(a_1,\ldots, a_n)\in \R^n$,
 \begin{align*}
       \P\big(X \leq a \vert A^\cmp\big)&= \P\Bigg( \bigcap_{i=1}^{n} \Big\{F^{-1}_{Y_i}\big( \ga U_i^c\big)\leq a_i\Big\} \,\Bigg \vert \, A^\cmp \Bigg) = \P\Bigg( \bigcap_{i=1}^{n} \Bigg\{   U_i^c \leq \frac{F_{Y_i}(a_i)\wedge \ga}{\ga}\Bigg\}\, \Bigg \vert \,A^\cmp \Bigg)  \\ 
       &= c\Bigg ( \frac{F_{Y_1}(a_1)\wedge \ga}{\ga} ,\ldots, \frac{F_{Y_n}(a_n)\wedge \ga}{\ga} \Bigg)\\
   &= c\big(\P(X_1\leq a_1\,|\, A^\cmp),\ldots, \P(X_n\leq a_n\,|\, A^\cmp)\big).
   \end{align*} 
 By definition $A$ is a common $\ga$-tail event for $X=(X_1,\ldots, X_n)$.
\end{proof}

The next corollary highlights an interesting consequence of Proposition \ref{th: cgammadominance} from the risk aggregation point of view.

\begin{corollary}\label{co: aggregationtail}
Let  $f\colon \R^n\to \R$ be non-decreasing and measurable, $c\colon [0,1]^n\to [0,1]$ be a copula, $\ga\in (0,1)$, $Y=(Y_1,\ldots, Y_n)$ be a random vector, and $A\in \cF$ be a $\ga$-tail event for $f(Y)$.\ Then, there exists a random vector $X=(X_1,\ldots, X_n)$ with $(c,\ga)$-hidden dependence, $X_i\overset{\d}{=}Y_i$ for $i=1,\ldots, n$ and common $\ga$-tail event $A$ such that
\[
f(X)\geq  f(Y) \quad  \P\text{-a.s.\ on }A.
\]
Thus, for any $\ga$-tail risk measure $R\colon \cX\to \R$ defined on a set of random variables $\cX$ containing $f(X)$ and $f(Y)$,
\[
R\big(f(X)\big)\geq R\big( f(Y)\big).
\]
\end{corollary}
\begin{proof}
 The first part of the statement follows directly from Proposition \ref{th: cgammadominance}, given that the aggregation function $f$ is non-decreasing.\ By Remark \ref{rem.tailevent.aggregate} $A$ is a $\ga$-tail event for $f(X)$ and, by Remark \ref{rem: tailriskandevents},  $\gamma$-tail risk measures only depend on the conditional distribution of random variables on a $\gamma$-tail event, which implies the second part of the statement.
\end{proof} 

Starting from Corollary \ref{co: aggregationtail}, we now focus on the study of the worst case scenario for a given tail risk measure under dependence uncertainty, i.e., when the risks' marginal distributions are given but their joint distribution is
not completely specified, see \eqref{eq: mainprob}.\ To this end, it is useful to formulate a compatibility assumption between the domain of the risk measure of interest and the considered aggregation function.
\begin{definition}\label{de.aggregation}
Let $f\colon \R^n\to \R$ and $\mu_1,\ldots, \mu_n\in \cP(\R)$.\ We say that a set $\cX$ of random variables is \textit{$(f,\mu_1,\ldots, \mu_n)$-aggregation-stable} if $f(X)\in \cX$ for all random vectors $X=(X_1,\ldots, X_n)$ with $X_i\sim\mu_i$ for $i=1,\ldots, n$.
\end{definition}

The following corollary states that when searching for the worst-case value attainable by any $\ga$-tail risk measure of a non-decreasing measurable aggregation function, one does not need to optimize over the set of all $n$-variate distributions compatible with the prescribed marginals, but the search can be narrowed down to the $n$-variate distributions of vectors with $(c,\ga)$-hidden dependence, where the copula $c\colon[0,1]^n\to [0,1]$ can be freely chosen.\ We stress that the statement is rather general as it includes arbitrary tail risk measures and aggregation functions. 

\begin{corollary} \label{co: full uncostrained}
           Let  $f\colon \R^n\to \R$ be non-decreasing and measurable, $\mu_1,\ldots, \mu_n\in \cP(\R)$, $c\colon [0,1]^n\to [0,1]$ be a copula, $\ga\in (0,1)$, and $\cX$ be an $(f,\mu_1,\ldots, \mu_n)$-aggregation-stable set of random variables.\ Then, for any $\ga$-tail risk measure $R\colon \cX\to \R$,  
\[\begin{aligned}
\sup \quad & {R}(f(X)) &= \quad \sup \quad & {R}(f(X))\\
\textrm{s.t.} \quad & X_i\sim \mu_i, \textbf{ }i=1,\ldots,n, & \textrm{s.t.} \quad & X_i\sim \mu_i, \textbf{ }i=1,\ldots,n.\\
& X \text{ has a $(c,\ga)$-hidden dependence} 
\end{aligned}\]
\end{corollary}

The following theorem is the first main result of this paper concerning the computation of risk bounds under dependence uncertainty.

\begin{theorem}
 \label{th.risk-aggregation-general}
           Let $f\colon \R^n\to \R$ be non-decreasing and measurable, $\mu_1,\ldots, \mu_n\in \cP(\R)$, $c\colon [0,1]^n\to [0,1]$ be a copula, $\ga\in (0,1)$, $\cX$ be an $(f,\mu_1,\ldots, \mu_n)$-aggregation-stable set of random variables, and assume that $\cC\subseteq \cP(\R^n)$ contains the $n$-variate distributions of all random vectors with $(c,\ga)$-hidden dependence and marginals $\mu_1,\ldots, \mu_n$.\ Then, for any $\ga$-tail risk measure $R\colon \cX\to \R$,
\[\begin{aligned}
\sup \quad & \mathcal{R}\big (f(X) \big)&= \quad \sup \quad & \mathcal{R}\big (f(X) \big)\\
\textrm{s.t.} \quad & X_i\sim \mu_i, \textbf{ }i=1,\ldots,n,  & \textrm{s.t.} \quad & X_i\sim\mu_i, \textbf{ }i=1,\ldots,n.\\
& \P_X\in \cC & &  \\
\end{aligned}\]
\end{theorem}

\begin{proof}
   Clearly, the left-hand side is less than or equal to the right-hand side.\ Since, by assumption, the set $\cC$ contains the $n$-variate distributions of all random vectors with $(c,\ga)$-hidden dependence and marginals $\mu_1,\ldots, \mu_n$,
\[\begin{aligned}
\sup \quad & \mathcal{R}\big (f(X) \big)& \geq \quad \sup \quad & \mathcal{R}\big (f(X) \big)\\
\textrm{s.t.} \quad & X_i\sim {\mu_i}, \textbf{ }i=1,\ldots,n,  & \textrm{s.t.} \quad & X_i\sim {\mu_i}, \textbf{ }i=1,\ldots,n,\\
& \P_X\in \cC & & X \text{ has a $(c,\ga)$-hidden dependence,}  \\
\end{aligned}\]
From Corollary \ref{co: full uncostrained} we know that 
\[\begin{aligned}
\sup \quad & {R}\big(f(X)\big) &= \quad \sup \quad & {R}\big(f(X)\big)\\
\textrm{s.t.} \quad & X_i\sim {\mu_i}, \textbf{ }i=1,\ldots,n, & \textrm{s.t.} \quad & X_i\sim {\mu_i}, \textbf{ }i=1,\ldots,n,\\
& X \text{ has a $(c,\ga)$-hidden dependence} &&    \\
\end{aligned}\]
which ends the proof.
\end{proof}

\begin{remark}
 Consider the situation of Theorem \ref{th.risk-aggregation-general} in the case of identically distributed marginals, i.e., $\mu_i=\mu\in \cP(\R)$ for $i=1,\ldots, n$, and $f(x)=\sum_{i=1}^n\la_i x_i$ with weights $\la_1,\ldots,\la_n\in [0,\infty)$.\ Then the set $\cX$ is $(f,\mu,\ldots,\mu)$-aggregation-stable if $\cX$ is a convex cone containing all $\mu$-distributed random variables.\ If $\cR\colon \cX\to \R$ is a sublinear $\ga$-tail risk measure, choosing $c^{\rm tail}$ in Lemma \ref{lem:char-hidden} as the co-monotonic copula, one obtains
\[\begin{aligned}
\sup \quad & \mathcal{R}\bigg(\sum_{i=1}^n \la_iX_i \bigg)&= \quad R(\mu) \sum_{i=1}^n \la_i, \\
\textrm{s.t.} \quad & X_i\sim \mu, \textbf{ }i=1,\ldots,n,  & &\\
& \P_X\in \cC & &  \\
\end{aligned}\]
where $\cR(\mu)$ denotes the risk $\cR(Z)$ of a $\mu$-distributed random variable $Z$.
\end{remark}

\subsection{Distributional Dependence Uncertainty}  \label{sec: partial}

In this section, we apply Theorem \ref{th.risk-aggregation-general} in the context of distributional dependence uncertainty.\ To that end, we introduce the concept of a consistent probability distance, which serves as the main tool to  describe partial information on the risks' dependence.

\begin{definition}
Let $\cM\subseteq \cP(\R)$ be nonempty.\ We say that a map $d\colon \cpl^n(\cM)\times \cpl^n(\cM)\to [0,\infty)$ is a \textit{consistent probability distance} {in $\pi\in \cpl^n(\cM)$} if
$$\lim_{k\to \infty} d(\pi_k,\pi)= 0$$
for all sequences $(\pi_k)_{k\in \N}\subset \cpl^n(\cM)$ with $\pi_k\to \pi$ as $k\to \infty$ and $\pi_k\circ \pr_i^{-1}=\pi\circ \pr_i^{-1}$ for all $k\in \N$ and $i=1,\ldots, n$.\ We say that $d$ is a \textit{consistent probability distance} if it is a consistent probability distance in every $\pi\in \cpl^n(\cM)$.
\end{definition}

Observe that, choosing $\pi^k=\pi$ in the previous definition, it follows that $d(\pi,\pi)=0$ for every consistent probability distance $d$. However, in contrast to a metric on a subset of probability measures, $d$ does not need to be symmetric nor does it have to satisfy the triangle inequality.\ We further point out that the term \textit{consistent} refers to consistency with weak convergence of probability measures.\ 

Let $\cM\subseteq \cP(\R)$.\ Again, we will fix a reference copula $c\colon [0,1]^n\to [0,1]$ and univariate distributions $\mu_1,\ldots,\mu_n\in \cM$.\ Then, distributional dependence uncertainty will be described via a consistent probability distance $d\colon \cpl^n(\cM)\times \cpl^n(\cM)\to [0,\infty)$ in $c(F_{\mu_1},\ldots, F_{\mu_n})$ and $\ep>0$, considering  the set
\begin{equation}\label{eq: uncertaintyset}
   \cC:= \big\{ \pi \in \cpl^n(\cM) \,\big|\, d\big(\pi,c(F_{\mu_1},\ldots, F_{\mu_n})\big)\leq \ep\big\}.
\end{equation}
This choice of $\cC$ collects all elements of $\cpl^n(\cM)$, which do not deviate more than $\ep >0$ from the reference distribution identified by the copula $c$ and the marginals $\mu_1,\ldots, \mu_n$ when the deviation is measured via the consistent probability distance $d$. 

Before we state our second main result, we provide a series of examples for consistent probability distances.
\begin{example}\label{ex.probabdistance}\
\begin{itemize}
\item[a)] Let $p\in [1,\infty)$ and $\cM$ be the set of all $\mu\in \cP(\R)$ with $\int_\R |x|^p\,\mu(\d x)<\infty$. Then, by \cite[Theorem 6.9]{villani2008optimal},
\[
d(\pi_1,\pi_2):=\cW^p(\pi_1,\pi_2),\quad\text{for }\pi_1,\pi_2\in \cpl^n(\cM),
\]
defines a consistent probability distance, where $\cW_p$ denotes the $p$-Wasserstein distance, cf.\ \cite[Chapter 6]{villani2008optimal} for a definition and elementary properties. 
\item[b)] Let $\cM=\cP(\R)$ and $d$ be an integral probability metric (IPM), i.e., 
$$
 d(\pi_1,\pi_2):= \sup_{f\in \mathscr F} \bigg| \int_{\R^n} f\,\d \pi_1-\int_{\R^n} f\,\d \pi_2\bigg|\quad \text{for all }\pi_1,\pi_2\in \cP(\R^n),
$$
where $\mathscr F$ is a non-empty set of measurable functions $\R^n\to \R$.\ Then $d$ is a consistent probability distance if the map
\begin{equation}\label{eq.dstar}
d^*\colon \R^n\times \R^n\to [0,\infty),\; (x,y)\mapsto \sup_{f\in \mathscr F} |f(x)-f(y)|
\end{equation}
is bounded and continuous,
cf. \cite[Corollary 2.7, p.\ 16]{zbMATH03517666} and \cite[Corollary 4.4]{zbMATH01112397}.
\item[c)]
 Let $\cM:=\cP(\R)$.\ Then, the Fortet-Mourier metric, i.e., the integral probability metric on $\cP(\R^n)$,  where $\mathscr F$ consists of all $1$-Lipschitz functions with $\sup_{x\in \R^n}|f(x)|\leq 1$ is a consistent probability distance.\ Moreover, it metrizes the weak topology on $\cP(\R^n)$, cf.\ \cite[Corollary 2.8, p.\ 17]{zbMATH03517666}.
 \item[d)] Let $\cM$ denote the set of all probability measures with continuous distribution function. Then, by \cite[Lemma 18]{CorradoMaxJan}, 
 \[
 d(\pi_1,\pi_2):= \sup_{a\in \R^n}\big| \pi_1\big((-\infty,a]\big) -\pi_2\big((-\infty, a]\big)\big|,\quad\text{for }\pi_1,\pi_2\in \cpl^n(\cM),
 \]
 defines a consistent probability distance.\ Observe that $d$ is also an IPM if one chooses $\mathscr F:=\{\eins_{(-\infty,a]}\, |\, a\in \R^n\}$.\ However, this choice of $\mathscr F$ does not lead to a continuous map $d^*$ given by \eqref{eq.dstar}.
 \item[e)] Let $\rho\colon \cpl^n(\cM)\to [-1,1]$ be a regular dependence measure in the sense of \cite{CorradoMaxJan}, i.e., $\rho(\mu_1\otimes \cdots\otimes \mu_n)=0$ for all $\mu_1,\ldots, \mu_n\in \cM$ and
 $\limsup_{k\to \infty} \rho(\pi_k)\leq 0$ for all sequences $(\pi_k)_{k\in \N}\subset \cpl^n(\cM)$ with $\pi_k\to \mu_1\otimes \cdots\otimes \mu_n$ as $k\to \infty$ and $\pi_k\circ \pr_i^{-1}=\mu_i\in \cM$ for all $k\in \N$ and $i=1,\ldots, n$.\ Then,
 \[
 d(\pi_1,\pi_2):=\max\big\{\rho(\pi_1)-\rho(\pi_2),0\big\},\quad \text{for }\pi_1,\pi_2\in \cpl^n(\cM),
 \]
 defines a consistent probability distance at $\mu_1\otimes \cdots\otimes \mu_n$ for all $\mu_1,\ldots, \mu_n\in \cM$.
\end{itemize}
\end{example}

\begin{remark}\label{rem: boundeddist}
 Let $d\colon \cP(\R^n)\times \cP(\R^n)\to [0,\infty)$ be an integral probability distance with a nonempty set $\mathscr F$ of measurable functions $\R^n\to \R$ such that \eqref{eq.dstar} is continuous and bounded by a constant $C\geq 0$.\ Moreover, let $Y=(Y_1,\ldots, Y_n)$ be a random vector with copula $c$, $\ga\in (0,1)$ and $X=(X_1,\ldots, X_n)$ be a random vector with $X_i\overset{\d}{=}Y_i$ for $i=1,\ldots, n$, $(c,\ga)$-hidden dependence, and common $\ga$-tail event $A\in \cF$. Then, by Lemma \ref{lem:char-hidden} (iv),
\begin{align*}
d(\P_X,\P_Y)&= \sup_{f\in \mathscr F}\Big|\E\big[f(X)\big]- \E\big[f(Y)\big]\Big|\leq \E\big[d^*(X,Y)\big]\\
&\leq \ga \E\big[d^*\big(F_Y^{-1}(\ga U),F_Y^{-1}(U)\big)\big] + C(1-\ga)\\
&= \ga \int_{[0,1]^n} d^*\big(F_Y^{-1}(\ga u),F_Y^{-1}(u)\big)\, \d c(u) + C(1-\ga).
\end{align*}
Since $F_Y^{-1}$ is left-continuous and $d^*$ is continuous and bounded, by the dominated convergence theorem, for all $\ep>0$, there exists some $\ga_0\in (0,1)$ such that
\[
d(\P_X,\P_Y)\leq \ep \quad \text{if}\quad \ga\geq \ga_0.
\] 
\end{remark}

The next theorem shows that no matter how small the deviation that we allow from the reference model identified by a copula $c\colon [0,1]^n \to [0,1]$ and prescribed marginals $\mu_1,\ldots, \mu_n\in \cM$ with a given nonempty set $\cM\subseteq\cP(\R)$, one can always find a $\gamma$ such that all coupling with
$(c,\gamma)$-hidden dependence belong to set of feasible models defined in \eqref{eq: uncertaintyset}.

\begin{theorem}
    \label{th: boundeddistance}
Let $\cM\subseteq \cP(\R)$ be nonempty, $\mu_1,\ldots, \mu_n\in \cM$, $c\colon [0,1]^n\to [0,1]$ be a copula, and $d\colon \cpl^n(\cM)\times \cpl^n(\cM)\to [0,\infty)$ be a consistent probability distance {at $c(F_{\mu_1},\ldots, F_{\mu_n})$}.\ Then, for all $\ep>0$, there exists $\ga\in (0,1)$ such that every random vector $X=(X_1,\ldots X_n)$ with $(c,\ga)$-hidden dependence and $X_i\sim \mu_i$ for $i=1,\ldots, n$ satisfies
$$d\big(\P_X,c(F_{\mu_1},\ldots, F_{\mu_n})\big)\leq \ep.$$
\end{theorem}
\begin{proof}
    Since the Fortet-Mourier metric $d_{\rm FM}$, cf.\ Example \ref{ex.probabdistance} c), metrizes the weak topology on $\cP(\R^n)$, for every $\ep>0$, there exists some $\de>0$ such that
    \[
     d_{\rm FM}\big(\P_X,c(F_{\mu_1},\ldots, F_{\mu_n})\big)\leq \de \quad\text{implies}\quad d\big(\P_X,c(F_{\mu_1},\ldots, F_{\mu_n})\big)\leq\ep.
    \]
    By Remark \ref{rem: boundeddist}, for all $\de>0$, there exists some $\ga\in (0,1)$ such that
    $$d_{\rm FM}\big(\P_X,c(F_{\mu_1},\ldots, F_{\mu_n})\big)\leq \de$$
   for every random vector $X=(X_1,\ldots X_n)$ with $(c,\ga)$-hidden dependence and $X_i\sim \mu_i$ for $i=1,\ldots, n$. The proof is complete.
\end{proof}

A combination of Theorem \ref{th.risk-aggregation-general} and Theorem \ref{th: boundeddistance} leads to our second main result for the computation of risk bounds.\ Namely, no matter how small the allowed perturbation from a given reference model is, there always exists $\gamma\in (0,1)$ such that the constraint expressed using a consistent probability distance does not help to reduce the risk measure worst-case scenario.

\begin{theorem}
 \label{th: constrainedwc}
           Let  $f\colon \R^n\to \R$ be non-decreasing and measurable, $\cM\subseteq \cP(\R)$ be nonempty, $\mu_1,\ldots, \mu_n\in \cM$, $c\colon [0,1]^n\to [0,1]$ be a copula, $\cX$ be an $(f,\mu_1,\ldots, \mu_n)$-aggregation-stable set of random variables, and $d\colon \cpl^n(\cM)\times \cpl^n(\cM)\to [0,\infty)$ be a consistent probability distance {at $c(F_{\mu_1},\ldots, F_{\mu_n})$}.\ Then, for all $\ep>0$, there exists $\ga\in (0,1)$ such that
\[\begin{aligned}
\sup \quad & \mathcal{R}\big (f(X) \big)&= \quad \sup \quad & \mathcal{R}\big (f(X) \big)\\
\textrm{s.t.} \quad & X_i\sim {\mu_i}, \textbf{ }i=1,\ldots,n,  & \textrm{s.t.} \quad & X_i\sim {\mu_i}, \textbf{ }i=1,\ldots,n\\
& d\big(\P_X,c(F_{\mu_1},\ldots, F_{\mu_n})\big)\leq \ep & &  \\
\end{aligned}\]
for every $\ga$-tail risk measure $\cR\colon \cX\to \R$.
\end{theorem}
\begin{proof}
   Fix $\ep>0$. We know from Theorem \ref{th: boundeddistance} that there exists $\gamma\in(0,1)$ such that every random vector $X=(X_1,\ldots X_n)$ with $(c,\ga)$-hidden dependence and $X_i\sim \mu_i$ for $i=1,\ldots, n$ satisfies
$d\big(\P_X,c(F_{\mu_1},\ldots, F_{\mu_n})\big)\leq \ep$.\ The claim now follows from Theorem \ref{th.risk-aggregation-general}.
\end{proof}

While Theorem \ref{th: constrainedwc} presents a purely qualitative result, explicit values for $\gamma$ depending on $\ep>0$ and vice versa are given in Section \ref{sec.wasserstein}, below, for Wasserstein distances and can be derived from the estimate in Remark \ref{rem: boundeddist} for a variety of integral probability metrics.

\subsection{The Case of Dependence Measures}\label{sec.dependence.meas}

Theorem \ref{th: constrainedwc} considers the worst-case aggregate risk in the case where the partial information concerning dependence is described using a ball around a reference $n$-variate distribution.\ The general idea is to consider all models that do not deviate too much from the reference coupling as plausible.\ An alternative approach that has been considered in the literature on model uncertainty in the context of risk aggregation is to describe partial dependence information via dependence measures, cf.\ \cite{bernard2023impact,Bernruvan2017,KAAS2009}.\ 
Since each contribution usually chooses a specific dependence measure to study the related risk aggregation problem, one can naturally wonder if it is possible to establish results on the risk measures' worst-case scenarios that hold regardless of the specific dependence measure under consideration. The answer provided by the following corollary is affirmative and obtained using the connection between consistent probability distances and a broad class of dependence measures that are described as regular dependence measures in \cite{CorradoMaxJan}.

Example \ref{ex.probabdistance} e) already highlighted a first connection between consistent probability distances and regular dependence measures.\ More generally, for a nonempty set $\cM\subseteq\cP(\R)$, we consider a map $\rho\colon \cpl^n(\cM)\to \R$ with $\rho(\pi)=\lim_{k\to \infty}\rho(\pi^k)$ for all sequences $(\pi^k)_{k\in \N}\subset \cpl^n(\cM)$ with $\pi^k\to \pi\in  \cpl^n(\cM)$ and $\pi^k\circ \pr_i^{-1}=\pi\circ \pr_i^{-1}$ for all $k\in \N$ and $i=1,\ldots, n$.\ This is a slightly stronger property than regularity for a dependence measure in the sense of \cite{CorradoMaxJan}.\ However, using similar arguments as in \cite[Example 3 and Example 4]{CorradoMaxJan}, all examples presented there, in particular, all dependence measures based on the Pearson correlation, Spearman's rho, Kendall's tau, and transport dependencies actually satisfy this stronger property.\ Then,
 \[
 d(\pi_1,\pi_2):=\big|\rho(\pi_1)-\rho(\pi_2)\big|,\quad \text{for }\pi_1,\pi_2\in \cpl^n(\cM),
 \]
 defines a consistent probability distance.\ The next corollary embeds Theorem \ref{th: constrainedwc} into the case in which partial dependence information is described via a dependence measure $\rho$, satisfying the previously mentioned properties without necessarily requiring the normalizations $\rho(\mu_1\otimes \cdots\otimes \mu_n)=0$ for all $\mu_1,\ldots, \mu_n\in \cM$ and $\rho(\pi)\in [-1,1]$ for all $\pi\in \cpl^n(\cM)$.
\begin{corollary}
 \label{co: constraineddepmeasures}
           Let  $f\colon \R^n\to \R$ be non-decreasing and measurable, $\cM\subseteq\cP(\R)$ be nonempty, $\mu_1,\ldots, \mu_n\in \cM$, $c\colon [0,1]^n\to [0,1]$ be a copula, $\cX$ be an $(f,\mu_1,\ldots, \mu_n)$-aggregation-stable set of random variables, and $\rho \colon \cpl^n(\cM)\to \R$ be a map with $\rho(\pi)=\lim_{k\to \infty}\rho(\pi^k)$ for all sequences $(\pi^k)_{k\in \N}\subset \cpl^n(\cM)$ with $\pi^k\to \pi=c(F_{\mu_1},\ldots,F_{\mu_n})$ and $\pi^k\circ \pr_i^{-1}=\mu_i$ for all $k\in \N$ and $i=1,\ldots, n$.\ Then, for all $\ep>0$, there exists $\ga\in (0,1)$ such that
\[\begin{aligned}
\sup \quad & \mathcal{R}\big (f(X) \big)&= \quad \sup \quad & \mathcal{R}\big (f(X) \big)\\
\textrm{s.t.} \quad & X_i\sim {\mu_i}, \textbf{ }i=1,\ldots,n,  & \textrm{s.t.} \quad & X_i\sim {\mu_i}, \textbf{ }i=1,\ldots,n.\\
& \big|\rho\big(\P_X\big)-\rho\big(c(F_{\mu_1},\ldots, F_{\mu_n})\big)\big|\leq \ep & &  \\
\end{aligned}\]
\end{corollary}

Again Corollary \ref{co: constraineddepmeasures} presents a purely qualitative result, and we refer to Section \ref{sec.dependendecomputation} for explicit analytic bounds and guarantees for $\ga\in (0,1)$ and $\ep>0$, respectively.

\section{Value at Risk and Expectiles under Hidden Dependence}\label{sec:varandexpectile}

The results obtained so far have highlighted the prominent role of $(c,\gamma)$-hidden dependence in the context of risk aggregation problems involving $\gamma$-tail risk measures. We now show that, in some cases, the same can be true also for risk measures beyond this class.\ In a first step, we study the left-continuous $\al$-VaR of the aggregate position $f(X)$ for a random vector $X=(X_1,\ldots, X_n)$ with $(c,\ga)$-hidden dependence and $0<\al\leq \ga<1$.\ We point out that, for such values of $\alpha$, the $\al$-VaR is \textit{not} a $\ga$-tail risk measure. The following lemma collects several properties of $\ga$-tail events and will be used in the proofs of this section.\ The statement in a) is implicitly given in \cite[Lemma A.3]{WangZitikis2020}, albeit in a slightly different form. Parts b) and c), which are closely tied, are given in \cite[Equation (11)]{WangZitikis2020} without a proof.\ For the sake of a self-contained exposition, we therefore state these results in our setup and provide a short proof.

\begin{lemma}\label{lem: tailevent}
    Let $Z$ be a random variable, $\ga\in (0,1)$, and $A\in \cF$ be a $\ga$-tail event for $Z$.\
    \begin{enumerate}
    \item[a)] It holds
    \begin{equation}\label{eq.estimate.esssup}
    \VaR^\ga(Z)=\esssup Z(A^\cmp)\leq \sup Z(A^\cmp),
    \end{equation}
where
\[
\esssup Z(A^\cmp):=\inf\big\{ a\in \R  \,\big|\, \P( Z> a\,|\, A^\cmp)=0\big\}.
\]
\item[b)] For all $a\in \R$, \begin{equation}\label{eq.cond.distribution}
 \P(Z\leq a\,|\, A^\cmp)=\frac{F_{Z}(a)\wedge\ga }{\ga}\quad\text{and}\quad \P(Z\leq a\,|\, A)= \frac{F_{Z}(a)-\ga}{1-\ga}\vee 0.
 \end{equation}
 \item[c)] For all $\al\in (0,1)$,
\begin{equation}\label{eq.cond.var}
  \VaR^{\al}(Z\,|\,A^\cmp)=F_Z^{-1}(\ga \al)\quad \text{and}\quad \VaR^{\al}(Z\,|\,A)=F_Z^{-1}\big(\ga+(1-\ga)\al\big).
 \end{equation}
\end{enumerate}
\end{lemma}
\begin{proof}\
\begin{enumerate}
\item[a)] Since $\{Z>\sup Z(A^\cmp)\}\cap A^\cmp=\emptyset$, it follows that 
\begin{equation}\label{eq.estimate.leftvar}
\esssup Z(A^\cmp)\leq \sup Z(A^\cmp).
\end{equation}
Now, let $a\in \R$ with $\P(Z>a\,|\, A^\cmp)=0$.\ Then,
\[
\P(Z>a)=\P\big(\{Z>a\}\cap A\big)\leq \P(A)=1-\ga,
\]
 so that, by definition of $\VaR^\ga (Z)$ and $\esssup Z(A^\cmp)$, we find that
 \[
\VaR^\ga(Z)\leq \esssup Z(A^\cmp)\leq \sup Z(A^\cmp).
 \]
Hence, by \eqref{eq.ga-tailevent} and the definition of $\VaR^\ga (Z)$, for $a\in \R$ with $a<\VaR^\ga(Z)$, 
 \[
  \ga \P(Z>a\,|\, A^\cmp)=\P\big(\{Z>a\}\setminus A\big)=\P(Z>a)-\P(A)= \P(Z>a)-(1-\ga)>0, 
 \]
which implies that
\[
\esssup Z(A^\cmp)\leq \VaR^\ga(Z).
\]
\item[b)] Let $a\in \R$.\ If $a\geq \VaR^\ga(Z)$, then $\P(Z\leq a\,|\, A^\cmp)=1$ by \eqref{eq.estimate.esssup}, so that
 \[
  \P\big(Z\leq a\,|\, A\big)=\frac{\P(Z\leq a)-\P\big(\{Z\leq a\}\cap A^\cmp\big)}{1-\ga}=\frac{F_Z(a)-\ga}{1-\ga}.
 \]
 On the other hand, if $a<\VaR^\ga(Z)$.\ Then, $Z(\om)\leq a$ implies $\om \in A^\cmp$ for all $\om\in \Om$, so that $\P(Z\leq a\,|\, A)=\frac{\P(\emptyset)}{1-\ga}=0$ and
 \[
 \P\big(Z\leq a\,|\, A^\cmp\big)=\frac{\P\big(\{Z\leq a\}\cap A^\cmp\big)}{\ga}=\frac{\P(Z\leq a)}{\ga}=\frac{F_Z(a)}{\ga}.
 \]
\item[c)] Let $\al\in (0,1)$.\ Then, by part b), for all $a\in \R$,
    \[
     \P(Z\leq a\,|\, A^\cmp)\geq \al\quad \text{if and only if}\quad F_Z(a)\geq \ga\al
    \]
    and
      \[
     \P(Z\leq a\,|\, A)\geq \al\quad \text{if and only if}\quad F_Z(a)\geq \ga+(1-\ga)\al.
    \]
 Hence,
 \[
  \VaR^{\al}(Z\,|\,A^\cmp)=\inf\big\{a\in \R\,\big|\, \P(Z\leq a\,|\, A^\cmp)\geq \al\big\}=\inf\big\{a\in \R\,\big|\, F_Z(a)\geq \ga \al\big\}=F_Z^{-1}(\ga \al)
 \]
 and
 \begin{align*}
 \VaR^{\al}(Z\,|\,A)&=\inf\big\{a\in \R\,\big|\, \P(Z\leq a\,|\, A)\geq \al\big\}=\inf\big\{a\in \R\,\big|\, F_Z(a)\geq \ga+(1-\ga) \al\big\}\\
 &=F_Z^{-1}\big(\ga+(1-\ga) \al\big).
 \end{align*}
 \end{enumerate}
\end{proof}

Recall that a survival copula $\overline{c}:[0,1]^n \to [0,1]$ of a vector $X=(X_1,\ldots,X_n)$ is a function that satisfies \[ \P(X>a)=\overline{c}\big(\P(X_1>a_1), \ldots,\P(X_n>a_n) \big)\]
for any $a \in \R^n$.\ We refer to \cite[Section 4.4]{zbMATH05080942} for a structured introduction to survival copulas and their applications in finance. Given a copula $c\colon [0,1]^n\to [0,1]$, its survival copula $\overline c$ can be expressed in terms of $c$ via
\begin{equation}\label{eq.survival}
 \overline c(v)=\int_{1-v_1}^1\cdots \int_{1-v_n}^1 \,c(\d u_1,\ldots, \d u_n)\quad \text{for }v=(v_1,\ldots, v_n)\in [0,1]^n,
\end{equation}
cf.\ \cite[Theorem 4.7]{zbMATH05080942}.

The next proposition shows that, if a vector $X$ has $(c,\gamma)$-hidden dependence and the corresponding survival copula satisfies $\overline c(u,\ldots, u)>0$ for all $u\in (0,1)$, the left-continuous $\ga$-Value-at-Risk VaR$^\gamma$ of the aggregate position $f(X)$ coincides with one attained in the case where the risks behave co-monotonically.\ 
\begin{proposition}\label{thm.comonotonic.main-1}
 Let $f\colon \R^n\to \R$ be non-decreasing and Borel measurable, $c\colon [0,1]^n\to [0,1]$ be a copula, $\ga\in (0,1)$, and $X=(X_1,\ldots, X_n)$ be a random vector with $(c,\ga)$-hidden dependence and common $\ga$-tail event $A\in \cF$.
 \begin{enumerate}
    \item[a)] For all $\al\in (0,\ga)$,
 \begin{equation}\label{eq.var-alpha-small}
    \VaR^\al\big(f(X)\big)=\VaR^{\frac\al\ga}\big( f(X)\,\big|\, A^\cmp\big).
     \end{equation}
    \item[b)] If $f$ is additionally left-continuous and $\overline c(u,\ldots, u)>0$ for all $u\in (0,1)$, then
     \[
  \VaR^\ga\big(f(X)\big)=f \big(\VaR^\ga (X_1),\ldots,\VaR^\ga (X_n)\big).
 \]
 \end{enumerate}
\end{proposition}

\begin{proof}
Part a) is a direct consequence of \eqref{eq.cond.var} together with Remark \ref{rem.tailevent.aggregate}.\ Let
\[
 q_X(\al):= f \big(\VaR^\al (X_1),\ldots,\VaR^\al(X_n) \big)
 \]
 for $\al\in (0,1)$.\ By part a) and \eqref{eq.estimate.esssup} together with the fact that $f$ is non-decreasing, it follows that
\[
\VaR^\ga \big(f(X)\big)=\lim_{\al\uparrow \ga} \VaR^\al \big(f(X)\big)=\lim_{\al\uparrow \ga}\VaR^{\frac\al\ga}\big(f(X)\,\big|\,A^\cmp\big)\leq q_X(\ga),
\]
and it remains to show that $q_X(\ga)\leq \VaR^\ga\big(f(X)\big)$.\ To that end, let $a\in \R$ with $a<q_X(\ga)$.\ Then, using the left-continuity of $q$, there exists some $\al\in (0,\ga)$ with $a< q_X(\al)\leq q_X(\ga)$.\ Hence, using the fact that $\overline c(u,\ldots, u)>0$ for all $u\in (0,1)$ together with part a),
\begin{align*}
\P\big(f(X)>a\,\big|\, A^\cmp\big)&\geq \P\big(f(X)\geq q_X(\al)\,\big|\, A^\cmp\big)\geq\P\bigg(\bigcap_{i=1}^n \big\{X_i\geq \VaR^\al(X_i)\big\} \,\bigg|\, A^ \cmp\bigg)\\
&\geq \overline c\Big(\P\big(X_1\geq \VaR^\al(X_1)\,\big|\, A^\cmp\big),\ldots, \P\big(X_n\geq \VaR^\al(X_n)\,\big|\, A^\cmp\big)\Big)\\
&= \overline c\Big(\P\big(X_1\geq \VaR^{\frac\al\ga}(X_1\,|\, A^\cmp)\,\big|\, A^\cmp\big),\ldots, \P\big(X_n\geq \VaR^{\frac\al\ga}(X_n\,|\, A^\cmp)\,\big|\, A^\cmp\big)\Big)\\
&\geq \overline c\bigg(1-\frac{\al}{\ga},\ldots, 1-\frac{\al}{\ga}\bigg)>0.
\end{align*}
By Remark \ref{rem.tailevent.aggregate} and \eqref{eq.estimate.esssup}, we thus obtain
 \[
    \VaR^\ga\big(f(X)\big)=\inf\big\{a\in \R\,\big|\, \P\big(f(X)>a\,\big|\, A^\cmp\big)=0\big\}\geq q_X(\ga).
 \]
 The proof is complete.
\end{proof}

We conclude this section with an application of Proposition \ref{th: cgammadominance} in the context of expectiles, cf.\ \cite{MR4216328}. For a random variable $Z$ on $\Om$ with $\E(|Z|)<\infty$ and $\al\in (0,1)$, the \textit{$\al$-expectile} ${\rm ex}^\al(Z)\in \R$ of $Z$ is defined as the unique solution to the equation
\[
\al\E\Big(\big(Z-{\rm ex}^\al(Z)\big)_+\Big)=(1-\al)\E\Big(\big(Z-{\rm ex}^\al(Z)\big)_-\Big).
\]
Note that, for $\al=\frac12$, the expectile coincides with the expectation, i.e., ${\rm ex}^{\frac12}(Z)=\E(Z)$.\ We stress that expectiles are law-invariant risk measures that have attracted much attention in the literature, but they are not $\ga$-tail risk measures for any $\gamma\in(0,1)$.\ We refer to \cite{bellini2017risk, delbaen2013expectiles,MR3479327} for a detailed and axiomatic study of expectiles as well as their interpretation in a financial context.\

In the sequel, for $\al\in (0,1)$, we use the notation ${\rm ex}^\al(\mu)$ and $\VaR^\al(\mu)$ for the $\al$-expectile ${\rm ex}^\al(Z)$ and the left $\al$-value at risk $\VaR^\al(Z)$ of a $\mu$-distributed random variable $Z$, respectively.\ As a consequence of \cite[Theorem 3]{MR4216328}, we have the following result.\ The proof follows a strategy similar to that of \cite[Theorem 16]{CorradoMaxJan}.

\begin{theorem}\label{thm: expectiles}
 Let $\mu\in \cP(\R)$ with $\int_{\R} |z|\,\mu(\d z)<\infty$, $c\colon [0,1]^n\to [0,1]$ be a copula, $\ga\in (0,1)$, and assume that $\cC\subseteq \cP(\R^n)$ contains the $n$-variate distributions of all random vectors with $(c,\ga)$-hidden dependence and identically $\mu$-distributed marginals.\ Then, for all $\la_1,\ldots, \la_n\in [0,\infty)$ and all $\al\in \big[\frac12,1\big)$ with ${\rm ex}^\al(\mu)\geq \VaR^\ga(\mu)$,
\[\begin{aligned}
\sup \quad & {\rm ex}^\al\bigg (\sum_{i=1}^n\la_i X_i \bigg)&= \quad  {\rm ex}^\al(\mu) \sum_{i=1}^n\la_i.\\
\textrm{s.t.} \quad & X_i\sim \mu, \textbf{ }i=1,\ldots,n,  & &\\
& \P_X\in \cC & &  \\
\end{aligned}\]
\end{theorem}

\begin{proof}
 For $\al=\frac12$ the statement is clear since, in this case, the $\al$-expectile is the mean.\ Let $\la_1,\ldots, \la_n\in [0,\infty)$ and $\al\in \big(\frac12,1\big)$ with ${\rm ex}^\al(\mu)\geq \VaR^\ga(\mu)$.\ Since $\al\geq \frac12$, it follows that ${\rm ex}^\al$ is sublinear, so that the left-hand side is less than or equal to the right-hand side.\ Let $X=(X_1,\ldots, X_n)$ be a random vector with $(c,\ga)$-hidden dependence, common $\ga$-tail event $A\in \cF$, and co-monotonic copula $c^{\rm tail}$ in Lemma \ref{lem:char-hidden}, i.e., $X_i=X_1$ $\P$-a.s.\ on $A$ for $i=1,\ldots, n$.\ We will prove that
 \[
  {\rm ex}^\al\bigg (\sum_{i=1}^n\la_i X_i \bigg)= \quad  {\rm ex}^\al(\mu) \sum_{i=1}^n\la_i.
 \]
 By \cite[Theorem 3]{MR4216328}, it remains to show that
\begin{equation}\label{eq.thm.expectiles}
  \P\Big(\big(\la_i X_i-{\rm ex}^\al(\la_iX_i)\big)\big(\la_j X_j-{\rm ex}^\al(\la_jX_j)\big)<0\Big)=0\quad\text{for }i,j=1,\ldots, n. 
 \end{equation}
 Since ${\rm ex}^\al$ is law-invariant and positively homogeneous, and $X_i\sim \mu$ for $i=1,\ldots, n$, this is equivalent to showing that
\[
\P\Big(\big(X_i-{\rm ex}^\al(\mu)\big)\big(X_j-{\rm ex}^\al(\mu)\big)<0\Big)=0 \quad\text{for }i,j=1,\ldots, n. 
\]
To that end, first observe that
\begin{align*}
\P\Big(\big(X_i-{\rm ex}^\al(\mu)\big)\big(X_j-{\rm ex}^\al(\mu)\big)<0\Big)&= \P\big(\{X_i<{\rm ex}^\al(\mu)\}\cap \{X_j>{\rm ex}^\al(\mu)\}\big)\\
&\quad +\P\big(\{X_i>{\rm ex}^\al(\mu)\}\cap \{X_j<{\rm ex}^\al(\mu)\}\big) 
\end{align*}
for $i,j=1,\ldots, n$.\ Hence, it remains to show that
\[
\P\big(\{X_i<{\rm ex}^\al(\mu)\}\cap \{X_j>{\rm ex}^\al(\mu)\}\big)=0\quad\text{for }i,j=1,\ldots, n.
\]
Since ${\rm ex}^\al(\mu)\geq \VaR^\ga(\mu)$, by \eqref{eq.estimate.esssup}, it follows that
\begin{align*}
\P\big(\{X_i<{\rm ex}^\al(\mu)\}\cap & \{X_j>{\rm ex}^\al(\mu)\}\big)=\P\big(\{X_i<{\rm ex}^\al(\mu)\}\cap \{X_j>{\rm ex}^\al(\mu)\}\cap A\big)\\
&\; =\P\big(\big\{X_1<{\rm ex}^\al(\mu)\big\}\cap \big\{X_1>{\rm ex}^\al(\mu)\big\}\cap A\big)=\P(\emptyset)=0
\end{align*}
for $i,j=1,\ldots, n$. We have therefore proved the validity of \eqref{eq.thm.expectiles}, and the statement now follows from \cite[Theorem 3]{MR4216328}.
\end{proof}
A combination of Theorem \ref{th: boundeddistance} and Theorem \ref{thm: expectiles} yields the following corollary.

\begin{corollary}
 \label{cor: expectiles}
           Let $\cM\subseteq \cP(\R)$ be nonempty, $\mu\in \cM$ with $\int_{\R} |z|\,\mu(\d z)<\infty$, $c\colon [0,1]^n\to [0,1]$ be a copula, and $d\colon \cpl^n(\cM)\times \cpl^n(\cM)\to [0,\infty)$ be a consistent probability distance {at $c(F_{\mu},\ldots, F_{\mu})$}.\ Then, for all $\ep>0$, there exists $\ga\in (0,1)$ such that
\[\begin{aligned}
\sup \quad & {\rm ex}^\al\bigg (\sum_{i=1}^n\la_i X_i \bigg)&= \quad  {\rm ex}^\al(\mu) \sum_{i=1}^n\la_i\\
\textrm{s.t.} \quad & X_i\sim \mu, \textbf{ }i=1,\ldots,n,  & &\\
& d\big(\P_X,c(F_\mu,\ldots, F_\mu)\big)\leq \ep & &  \\
\end{aligned}\]
for all $\la_1,\ldots, \la_n\in [0,\infty)$ and all $\al\in \big[\frac12,1\big)$ with ${\rm ex}^\al(\mu)\geq \VaR^\ga(\mu)$.
\end{corollary}

\section{Computational Aspects}

The results obtained in the previous sections are mainly of a qualitative nature.\ In this section, we turn our focus to quantifying the interplay between $\ga\in (0,1)$, the sensitivity level $\ep>0$, and the maximal tail risk in specific setups from a practical risk management perspective.

\subsection{Dependence Measures}\label{sec.dependendecomputation}
In this subsection, we carry out an analysis for the relationship between the sensitivity level $\ep>0$ and $\ga\in (0,1)$ in the setup from Section \ref{sec.dependence.meas} where we studied risk bounds in the case of dependence uncertainty described by regular dependence measures.\ Throughout this subsection, let $Y=(Y_1,\ldots, Y_n)$ be a random vector with copula $c\colon [0,1]^n\to [0,1]$, $\ga\in (0,1)$ and $X=(X_1,\ldots, X_n)$ be a random vector with $X_i\overset{\d}{=}Y_i$ for $i=1,\ldots, n$, $(c,\ga)$-hidden dependence, and common $\ga$-tail event $A\in \cF$.\ For $i,j=1,\ldots, n$, we denote by $c_{ij}$ and $c_{ij}^{\rm tail}$ the two-dimensional distribution of $(U_i,U_j)$ and $(V_i,V_j)$ for $U=(U_1,\ldots, U_n)\sim c$ and $U=(V_1,\ldots, V_n)\sim c^{\rm tail}$, conditionally on $A^\cmp$ and $A$, respectively, where $c^{\rm tail}$ is an arbitrary but fixed copula, see Lemma \ref{lem:char-hidden} (iv).

We start with the case in which the dependence measure is given by the Pearson correlation.

\begin{example}[Pearson Correlation]\label{ex.pearson}
     Let $i,j=1,\ldots, n$.\ Then,
     \[
     \big|\cor (Y_i,Y_j)-\cor (X_i,X_j)\big|=\frac{ \big|\E [Y_iY_j]-\E[X_iX_j]\big|}{\sqrt{\var(Y_i)\var(Y_j)}}.
     \]     
     By Lemma \ref{lem:char-hidden} (iv), 
    \begin{align}
\notag  \big|\E [X_iX_j]-\E[Y_iY_j]\big|&= \Big|\ga \E \big[F_{Y_i}^{-1}(\ga U_i)F_{Y_j}^{-1}(\ga U_j)\,\big|\, A^\cmp\big]-\E \big[F_{Y_i}^{-1}(U_i)F_{Y_j}^{-1}(U_j)\,\big|\, A^\cmp\big]\\
 &\quad +(1-\ga)\E \big[F_{Y_i}^{-1}\big(\ga +(1-\ga)V_i\big)F_{Y_j}^{-1}\big(\ga+(1-\ga) V_j\big)\,\big|\, A\big]\Big|. \label{eq.pearson}
    \end{align}
 In the sequel, we use equation \eqref{eq.pearson} to determine $\ga\in (0,1)$ such that 
 \[
 \big|\cor(X_i,X_j)-\cor(Y_i,Y_j)\big|\leq \ep
 \]
 for $\ep>0$, independently of the copula $c^{\rm tail}$, for Bernoulli and exponentially distributed marginals.
  \begin{enumerate}
    \item[(i)] Let $Y_i\deq Y_j\sim B(1,p)$ be Bernoulli distributed with $p\in (0,1)$ and $\ga>1-p$.\ Then, $F_{Y_i}^{-1}=F_{Y_j}^{-1}=\eins_{[1-p,1)}$ and equation \eqref{eq.pearson} simplifies to
        \begin{align*}
  \big|\E [X_iX_j]-\E[Y_iY_j]\big|&= \Big|\ga \P\big(\{\ga U_i> 1-p\}\cap\{\ga U_j> 1-p\}\big)\\
  &\quad\quad-\P\big(\{U_i> 1-p\}\cap\{U_j> 1-p\}\big)+(1-\ga)\Big|\\
  &=  \Big|\ga {c_{ij}}\Big(\tfrac{1-p}{\ga},\tfrac{1-p}{\ga}\Big)-{c_{ij}}(1-p,1-p)\Big|,
    \end{align*}
    where, in the second step, we used the identity $$\P\big(\{U_i> u\}\cap\{U_j> u\}\big)=1-2u+c_{ij}(u,u)\quad\text{for all }u\in [0,1],$$ cf.\ \cite[Section 2.5, p.\ 75]{zbMATH05080942}.\
    In the case, where $c_{ij}(u,v)=uv$ for all $u,v\in [0,1]$, we thus obtain
    \[
     \big|\E [X_iX_j]-\E[Y_iY_j]\big|=\frac1\ga(1-p)^2-(1-p)^2.
    \]
    Hence, for $\ep>0$,
    \[
    \big|\cor(X_i,X_j)-\cor(Y_i,Y_j)\big|\leq \ep\quad \text{if and only if}\quad \ga\geq \frac{1}{1+\ep \frac{p}{1-p}}.
    \]
   \item[(ii)] Now, assume that $Y_i\deq Y_j\sim{\rm Exp}(\la)$ with $\la>0$, so that $$F_{Y_i}^{-1}(u)=F_{Y_j}^{-1}(u)=-\frac{\ln(1-u)}\la \quad \text{for }u\in (0,1).$$
   Moreover, let $Z_i:= F_{Y_i}^{-1}(V_i)$, $Z_j:= F_{Y_j}^{-1}(V_j)$, and $c_{ij}(u,v)=uv$ for all $u,v\in [0,1]$.\ Then, equation \eqref{eq.pearson} yields
 \begin{align*}
   \big|\cor(X_i,X_j)-\cor(Y_i,Y_j)\big| &=\Bigg|\frac1\ga \bigg(\int_0^\ga \ln(1-u)\,\d u\bigg)^2+(1-\ga)\big(\ln(1-\ga)\big)^2\\
   &\quad\quad -2(1-\ga)\ln(1-\ga) -\ga +(1-\ga)\cor(Z_i,Z_j)\Bigg|\\
   &=\Bigg|\frac{\big((1-\ga)\ln(1-\ga)+\ga\big)^2}{\ga}+(1-\ga)\big(\ln(1-\ga)\big)^2\\
   &\quad\quad-2(1-\ga)\ln(1-\ga) -\ga +(1-\ga)\cor(Z_i,Z_j)\Bigg|\\
   &= (1-\ga)\Bigg|\cor(Z_i,Z_j)+\frac{\big(\ln(1-\ga)\big)^2}{\ga}\Bigg|\\
   &\leq (1-\ga)\Bigg(1+\frac{\big(\ln(1-\ga)\big)^2}{\ga}\Bigg),
    \end{align*}
    which leads to 
    \[
    \big|\cor(X_i,X_j)-\cor(Y_i,Y_j)\big|\leq \ep \quad \text{if}\quad (1-\ga)\Bigg(1+\frac{\big(\ln(1-\ga)\big)^2}{\ga}\Bigg)\leq \ep
    \]
    for all $\ep>0$, independently of the copula $c^{\rm tail}$, see Figure \ref{fig:exponential_cor}.
    \end{enumerate}
       \begin{figure}[h]
\centering
\begin{tikzpicture}
\begin{axis}[
    axis lines = middle,
    xlabel = $\ga$,
    ylabel = $\ep$,
    domain=0:1,
    samples=200,
    xmin=0, xmax=1.1,
    ymin=0, ymax=1.1,
]
\addplot[black, thick]
{(1-x)*(1+(ln(1-x))^2/x)};
\end{axis}
\end{tikzpicture}
\caption{Correlation bounds as a function of $\ga\in (0,1)$ assuming $X_i\sim{\rm Exp}(\la)$ for $i=1,\ldots, n$ with $\la>0$, cf.\ Example \ref{ex.pearson} (ii).}
\label{fig:exponential_cor}
\end{figure}
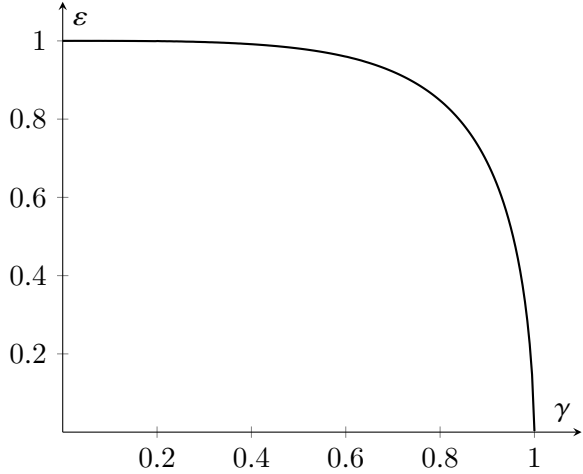
 \end{example}

Next, we consider Spearman's rho as a dependence measure.\ Since Spearman's rho is independent of the marginals, we are able to come up with explicit bounds as the following example shows.
 \begin{example}[Spearman's Rho]\label{ex.spearman} 
 Let $i,j=1,\ldots, n$.\ Using \eqref{eq.pearson} with $F_{Y_i}^{-1}(u)=F_{Y_j}^{-1}(u)=u$ for $u\in (0,1)$, we obtain
   \begin{align*}
    \big|\spear(X_i,X_j)- \spear(Y_i,Y_j)\big|&= 12\Big|\ga^3\E\big[U_iU_j\,\big|\,A^\cmp\big]+\ga(1-\ga)+(1-\ga)^3\E\big[V_iV_j\,\big|\,A\big]\\
    &\qquad\qquad\qquad\qquad-\E\big[U_iU_j\,\big|\,A^\cmp\big]\Big|\\ 
    &=\Big|\big(1-\ga^3\big) \big(1-\spear(c_{ij})\big)-(1-\ga)^3 \Big(1-\spear\big(c_{ij}^{\rm tail}\big)\Big)\Big|\\
    &=\Big| 3\ga(1-\ga) \big(1-\spear(c_{ij})\big)+(1-\ga)^3\Big(\spear\big(c_{ij}^{\rm tail}\big)-\spear(c_{ij})\Big)\Big|\\
    &\leq (1-\ga)\Big( \big|3\ga -\big(3\ga+(1-\ga)^2\big)\spear(c_{ij})\big|+(1-\ga)^2\Big),
    \end{align*}
    where, in the last step, we used the fact that Spearman's rho takes values in $[-1,1]$.\ We point out that this is a sharp bound, which is attained for $\spear\big(c_{ij}^{\rm tail}\big)=\pm 1$.\ Hence, if $\spear(c_{ij})\leq \frac{3\ga}{3\ga+(1-\ga)^2}$ then, for all $\ep>0$ and any choice of  $c^{\rm tail}$,
    \[
    \big|\spear(X_i,X_j)- \spear(Y_i,Y_j)\big|\leq \ep \quad\text{if}\quad \ga \geq \sqrt[3]{1-\frac{\ep}{1-\spear(c_{ij})}}.
    \]
\end{example}

As for Spearman's rho, also in the case of Kendall's tau, we obtain explicit bounds for the distance between $\P_X$ and $\P_Y$.

\begin{example}[Kendall's Tau]\label{ex.kendall} Let $i,j=1,\ldots, n$.\ Using Lemma \ref{lem:char-hidden} (iii) and (iv), we find that
\begin{align*}
    \big|\tau(X_i,X_j)-\tau(Y_i,Y_j)\big|&\leq 4\Big| \ga^2\E\big[c_{ij}(U_i,U_j)\, \big|\, A^\cmp\big] +\ga(1-\ga)+(1-\ga)^2\E\big[c_{ij}^{\rm tail}(V_i,V_j)\, \big|\, A\big]\\
    &\qquad\qquad\qquad - \E\big[c_{ij}(U_i, U_j)\, \big|\, A^\cmp\big]\Big|\\
    & = \Big|\big(1-\ga^2\big) \big(1-\tau(c_{ij})\big)-(1-\ga)^2 \Big(1-\tau\big(c_{ij}^{\rm tail}\big)\Big)\Big|\\
       & = \Big|2\ga(1-\ga) \big(1-\tau(c_{ij})\big)+(1-\ga)^2 \Big(\tau\big(c_{ij}^{\rm tail}\big)-\tau(c_{ij})\Big)\Big|\\
    &\leq (1-\ga)\big( \big|2\ga -(1+\ga)\tau(c_{ij})\big|+(1-\ga)\big),
\end{align*}
where, in the last step, we used the fact that Kendall's tau takes values in $[-1,1]$.\ We point out that this is again a sharp bound, which is attained for $\tau\big(c_{ij}^{\rm tail}\big)=\pm 1$.\
Hence, if $\tau(c_{ij})\leq \frac{2\ga}{1+\ga}$ then, for all $\ep>0$ and any choice of  $c^{\rm tail}$,
    \[
    \big|\tau(X_i,X_j)- \tau(Y_i,Y_j)\big|\leq \ep \quad\text{if}\quad \ga \geq \sqrt{1-\frac{\ep}{1-\tau(c_{ij})}}.
    \]
\end{example}

 \subsection{Wasserstein Distance for Hidden Dependence}\label{sec.wasserstein}
Theorem \ref{th: constrainedwc} states that, for every copula $c$ and $\ep>0$, there exists $\ga\in(0,1)$ such that the worst case scenario of $\ga$-tail risk measures is not affected by a dependence constraint of the form $d\big(\P_X,c(F_{\mu_1},\ldots, F_{\mu_n})\big)\leq \ep$, in which $d$ is a consistent probability distance.\ A practical implementation of such a result requires studying the relationship between $\ep$ and $\ga$.\ In this subsection, we study this relationship in the case where the consistent probability distance $d$ is given by the $p$-Wasserstein distance with $p\in [1,\infty)$, cf.\ Example \ref{ex.probabdistance} a).\ The analysis shows that for every $\ga \in (0,1)$ and every copula $c$, the $p$-Wasserstein distance attainable considering $(c,\ga)$-hidden dependence structures admits analytically tractable upper bounds.\ For the Fortet-Mourier distance and, more generally, a broad class of IPMs, this relation has already been illustrated in Remark \ref{rem: boundeddist}.\ We point out that, for $p=1$, the Wasserstein distance is an IPM with $d^*(x,y)=|x-y|$.\ In this case, we shall obtain a refined estimate compared to Remark \ref{rem: boundeddist}.\

Throughout this section, rather than using the usual Euclidean distance on $\R^n$, we choose the $\ell_p$-norm divided by $n^{1/p}$ in the definition of the Wasserstein distance, i.e.,
\[
|x|:=
\bigg(\frac{1}n\sum_{i=1}^n |x_i|^p\bigg)^{1/p}\quad \text{for }x\in \R^n.
\]
Let $Y=(Y_1,\ldots, Y_n)$ be a random vector with copula $c\colon [0,1]^n\to [0,1]$, $\ga\in (0,1)$ and $X=(X_1,\ldots, X_n)$ be a random vector with $X_i\overset{\d}{=}Y_i$ for $i=1,\ldots, n$, $(c,\ga)$-hidden dependence, and common $\ga$-tail event $A\in \cF$.\ Furthermore, assume that $\E[|Y_i|^p]<\infty$ for $i=1,\ldots, n$.\ Then, using Lemma \ref{lem:char-hidden} (iv),
\begin{align*}
 \cW_p(\P_X,\P_Y)^p&\leq\frac{1}n \sum_{i=1}^n \ga\int_0^1 \big(F_{Y_i}^{-1} (u)-F_{Y_i}^{-1}(\ga u)\big)^p\,\d u\\
 &\quad +(1-\ga) \int_0^1\int_0^1 \big|F_{Y_i}^{-1}\big(\ga+(1-\ga)v\big)-F_{Y_i}^{-1} (u)\big|^p\, \d u\, \d v.
\end{align*}
In the cases $p=1$ and $p=2$, the previous estimate can be made more explicit.\ We start with the case $p=1$.\ Then, using the previous estimate,
\begin{align*}
\cW_1(\P_X,\P_Y)&\leq  \frac1n\sum_{i=1}^n\ga \E[Y_i]-\int_0^\ga F_{Y_i}^{-1}(u)\, \d u + \int_\ga^1\int_0^1 \big|F_{Y_i}^{-1}(v)-F_{Y_i}^{-1}(u)\big|\, \d u\, \d v\\
&= \frac1n\sum_{i=1}^n\ga \E[Y_i]-\int_0^\ga F_{Y_i}^{-1}(u)\, \d u + \int_\ga^1\int_0^v F_{Y_i}^{-1}(v)-F_{Y_i}^{-1}(u)\, \d u\, \d v\\
&\quad +\int_\ga^1\int_v^1 F_{Y_i}^{-1}(u)-F_{Y_i}^{-1}(v)\, \d u\, \d v.
\end{align*}
We first focus on the second-to-last integral on the right-hand side.\ Using Fubini's theorem, we obtain
\begin{align*}
\int_\ga^1\int_0^v F_{Y_i}^{-1}(v)-F_{Y_i}^{-1}(u)\, \d u\, \d v&= \int_{\ga}^1 vF_{Y_i}^{-1}(v)\, \d v-(1-\ga)\int_0^\ga F_{Y_i}^{-1}(u)\, \d u\\
&\quad - \int_\ga^1 (1-u) F_{Y_i}^{-1}(u)\, \d u\\
&=\int_\ga^1 (2u-1) F_{Y_i}^{-1}(u)\, \d u-(1-\ga)\int_0^\ga F_{Y_i}^{-1}(u)\, \d u.
\end{align*}
Using again Fubini's theorem, for the last integral on the right-hand side, we find
\begin{align*}
\int_\ga^1\int_v^1 F_{Y_i}^{-1}(u)-F_{Y_i}^{-1}(v)\, \d u\, \d v&= \int_\ga^1 (u-\ga)F_{Y_i}^{-1}(u)\, \d u-\int_{\ga}^1 (1-v)F_{Y_i}^{-1}(v)\, \d v\\
& = \int_\ga^1 (2u-\ga -1) F_{Y_i}^{-1}(u)\, \d u.
\end{align*}
Altogether, using the substitution $s=u^2$, we thus get
\begin{align*}
\cW_1(\P_X,\P_Y)&\leq  \frac2n\sum_{i=1}^n\int_\ga^1 2uF_{Y_i}^{-1}(u)\, \d u-\ga(1-\ga) \TVaR^\ga (Y_i) -(1-\ga) \E[Y_i]\\
&=\frac2n\sum_{i=1}^n\int_{\ga^2}^1 F_{Y_i}^{-1}\big(\sqrt s\big)\, \d s-\ga(1-\ga) \TVaR\ga (Y_i) -(1-\ga) \E[Y_i]\\
&=\frac2n\sum_{i=1}^n(1-\ga^2)\TVaR^{\ga^2}(M_i)-\ga(1-\ga) \TVaR\ga (Y_i) -(1-\ga) \E[Y_i],
\end{align*}
where $$M_i:= F_{Y_i}^{-1}\big(\sqrt U\big)\quad \text{for }i=1,\ldots, n\text{ with }U\text{ uniform on }(0,1).$$\ Observe that, for $i=1,\ldots, n$, $M_i\deq Y_i\vee\widetilde Y_i$, where $\widetilde Y_i$ is an  independent copy of $Y_i$, since
\[
\P\big(Y_i\vee\widetilde Y_i\leq a\big)=\P\big(\{Y_i\leq a\}\cap \big\{\widetilde Y_i\leq a\big\}\big)=\P(Y_i\leq a)\cdot \P\big(\widetilde Y_i\leq a\big)=\big[F_{Y_i}(a)\big]^2
\]
for all $a\in \R$, which implies $F_{Y_i\vee\widetilde Y_i}^{-1}(u)=F_{Y_i}^{-1}\big(\sqrt{u}\big)$ for all $u\in (0,1)$.

For the case $p=2$, we get
\[
\cW_2(\P_X,\P_Y)^2\leq \frac2n\sum_{i=1}^n \E[Y_i^2]-\ga \int_0^1 F_{Y_i}^{-1}(u)F_{Y_i}^{-1}(\ga u)\, \d u-(1-\ga )\E[Y_i]\TVaR\ga(Y_i).
\]
Assume that $Y_i\geq 0$ $\P$-a.s.\ for all $i=1,\ldots, n$.\ Then, using the estimate $F_{Y_i}^{-1}(u)\geq F_{Y_i}^{-1}(\ga u)$ for $u\in (0,1)$, we get
\begin{align*}
 \cW_2(\P_X,\P_Y)^2&\leq \frac2n\sum_{i=1}^n \E[Y_i^2]-\ga \int_0^1 \big[F_{Y_i}^{-1}(\ga u)\big]^2\, \d u-(1-\ga )\E[Y_i]\TVaR^\ga(Y_i)\\
 &= \frac2n\sum_{i=1}^n \E[Y_i^2]-\int_0^\ga \big[F_{Y_i}^{-1}(u)\big]^2\, \d u-(1-\ga )\E[Y_i]\TVaR^\ga(Y_i)\\ 
 &=\frac{2(1-\ga)}n\sum_{i=1}^n \TVaR^\ga(Y_i^2)-\E[Y_i]\TVaR^\ga(Y_i).
\end{align*}

We conclude this subsection with a series of examples in which we are able to derive closed formulas for the upper bounds both for the $1$-Wasserstein and the $2$-Wasserstein distance.

\begin{example}\label{ex.wass1}\
In this example, we provide explicit bounds for the $1$-Wasserstein distance $\cW_1$ as a function of $\ga$ for marginals that are uniform, exponential or Bernoulli.
\begin{enumerate}
\item[a)] Let $a,b\in \R$ with $a<b$ and consider the case, where $Y_i$ is uniform on the interval $(a,b)$ for $i=1,\ldots, n$.\ Then, $\E[Y_i]=\frac{a+b}{2}$ and $F_{Y_i}^{-1}(u)=a+(b-a)u$ for $u\in (0,1)$ and $i=1,\ldots, n$.\ Hence,
\begin{align*}
\cW_1(\P_X,\P_Y)&\leq\Big(2a(1-\ga^2)+(b-a)(1-\ga^3)-2\ga(1-\ga)a\\
&\quad -(b-a)\ga (1-\ga^2)-(1-\ga)(a+b)\Big)+\frac{b-a}{3}(1-\ga^3)\\
&=\frac{b-a}{3}(1-\ga^3).
\end{align*}
Hence,
\[
\cW_1(\P_X,\P_Y)\leq \ep\quad \text{if}\quad \ga\geq \sqrt[3]{1-\frac{3\ep}{b-a}}.
\]
\item[b)] Let $\la>0$ and $Y_i\sim{\rm Exp}(\la)$ for $i=1,\ldots, n$.\ Then, $\E[Y_i]=\frac1\la$ and $F_{Y_i}^{-1}(u)=-\frac1\la \log (1-u)$ for $u\in (0,1)$ and $i=1,\ldots, n$.\ Hence,
\begin{align*}
\cW_1(\P_X,\P_Y)&\leq\frac1\la \Big(4(1-\ga)-(1-\ga)^2-\big(4(1-\ga)-2(1-\ga)^2\big)\log (1-\ga)\\
&\quad -2\ga(1-\ga)+2\ga(1-\ga)\log (1-\ga)-2(1-\ga)\Big)\\
&=\frac{(1-\ga)^2-2(1-\ga)\log(1-\ga)}\la.
\end{align*}
Hence,
\[
\cW_1(\P_X,\P_Y)\leq \ep\quad \text{if}\quad (1-\ga)^2-2(1-\ga)\log(1-\ga)\leq \la \ep,
\]
see Figure \ref{fig:exponential_wass1}.

\begin{figure}[h]
\centering
\begin{tikzpicture}
\begin{axis}[
    axis lines = middle,
    xlabel = $\ga$,
    ylabel = $\frac\ep\la$,
    domain=0:1,
    samples=200,
    xmin=0, xmax=1.1,
    ymin=0, ymax=1.15,
]
\addplot[black, thick]
{(1-x)^2 - 2*(1-x)*ln(1-x)};
\end{axis}
\end{tikzpicture}
\caption{$\cW_1$-bounds as a function of $\ga\in (0,1)$ assuming $X_i\sim{\rm Exp}(\la)$ for $i=1,\ldots, n$ with $\la>0$, cf.\ Example \ref{ex.wass1} b).}
\label{fig:exponential_wass1}
\end{figure}

\item[c)] Let $p\in(0,1)$ and $Y_i\sim{\rm B}(1,p)$ for $i=1,\ldots, n$.\ Then, $\E[Y_i]=p$ and $F_{Y_i}^{-1}(u)=\eins_{(1-p,1)}(u)$ for $u\in (0,1)$ and $i=1,\ldots, n$.\ Hence,
\begin{align*}
\cW_1(\P_X,\P_Y)&\leq2\Big(1-\max\{\ga,1-p\}^2-\ga\big(1-\max\{\ga,1-p\}\big)-(1-\ga)p\Big)\\
&=2(1-p)\min\{1-\ga,p\}.
\end{align*}
In the case where $1-\ga\leq p$, we thus find that
\[
\cW_1(\P_X,\P_Y)\leq \ep\quad \text{if}\quad \ga\geq 1-\frac{\ep}{2(1-p)}.
\]
\end{enumerate}
\end{example}
\begin{example}\label{ex.wass2}\
We now derive explicit bounds for the $2$-Wasserstein distance $\cW_2$ as a function of $\ga\in (0,1)$.\ Again, we consider marginal distributions that are uniform, exponential, or Bernoulli.
\begin{enumerate}
\item[a)] Let $a,b\in \R$ with $0\leq a<b$ and consider the case, where $Y_i$ is uniform on the interval $(a,b)$ for $i=1,\ldots, n$.\ Then, for $i=1,\ldots, n$, $\E[Y_i]=\frac{a+b}{2}$, $\TVaR^\ga(Y_i)= \frac{a+b+\ga(b-a)}{2}$, and
\[
\TVaR^\ga(Y_i^2)= \frac{b^2+\big(a+\ga(b-a)\big)\big(a+b+\ga(b-a)\big)}{3}.
\]
Hence,
\[
\cW_2(\P_X,\P_Y)^2\leq 
\frac{b-a}{6}
\Big(\big(1+3\gamma^2-4\gamma^3\big)b
+\big(5\gamma-1-4\gamma^2\big)a\Big).
\]
If $a=0$, this simplifies to 
\[
 \cW_2(\P_X,\P_Y)^2\leq \frac{b^2}6\big(1+3\ga^2-4\ga^3\big),
\]
so that, in this case, 
\[
\cW_2(\P_X,\P_Y)\leq \ep\quad \text{if}\quad 1+3\ga^2-4\ga^3\leq \frac{6\ep^2}{b^2},
\]
see Figure \ref{fig:uniform}.

\begin{figure}[h]
\centering
\begin{tikzpicture}
\begin{axis}[
    axis lines = middle,
    xlabel = $\ga$,
    ylabel = $b^2\ep^2$,
    style={
        /pgf/number format/.cd,
        fixed,
        precision=2,
    },
    domain=0:1,
    samples=200,
    xmin=0, xmax=1.1,
    ymin=0, ymax=0.2,
]
\addplot[black, thick]
{((1+3*x^3-4*x^4)/6)};
\end{axis}
\end{tikzpicture}
\caption{$\cW_2^2$-bounds as a function of $\ga\in (0,1)$ assuming $X_i\sim \mathcal{U}(0,b)$ for $i=1,\ldots, n$ with $b>0$, cf.\ Example \ref{ex.wass2} a).}
\label{fig:uniform}
\end{figure}
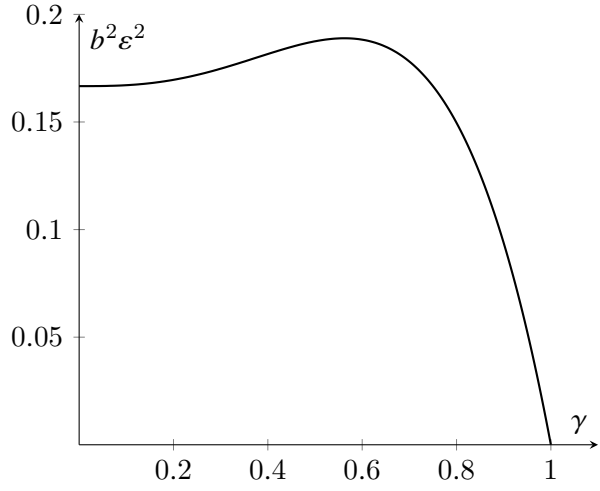

\item[b)] Let $\la>0$ and $Y_i\sim{\rm Exp}(\la)$ for $i=1,\ldots, n$.\ Then, for $i=1,\ldots, n$, $\E[Y_i]=\frac1\la$, $\TVaR^\ga(Y_i)=\frac{1-\log (1-\ga)}{\la}$, and
\[
\TVaR^\ga(Y_i^2)=\frac{2+[\log(1-\ga)]^2-2\log(1-\ga)}{\la^2}.
\]
Hence,
\[
\cW_2(\P_X,\P_Y)^2\leq 2(1-\ga)\frac{1+[\log(1-\ga)]^2-\log(1-\ga)}{\la^2},
\]
so that
\[
\cW_2(\P_X,\P_Y)\leq \ep\quad \text{if}\quad (1-\ga)\big(1+[\log(1-\ga)]^2-\log(1-\ga)\big)\leq \frac{\la^2 \ep^2}2,
\]
see Figure \ref{fig:exponential_wass2}.

\begin{figure}[h]
\centering
\begin{tikzpicture}
\begin{axis}[
    axis lines = middle,
    xlabel = $\ga$,
    ylabel = $\frac{\ep^2}{\la^2}$,
    domain=0:1,
    samples=200,
    xmin=0, xmax=1.1,
    ymin=0, ymax=2.35,
]
\addplot[black, thick]
{(2*(1-x)*(1+(ln(1-x))^2 - ln(1-x)))};
\end{axis}
\end{tikzpicture}
\caption{$\cW_2^2$-bounds as a function of $\ga\in (0,1)$ assuming $X_i\sim{\rm Exp}(\la)$ for $i=1,\ldots, n$ with $\la>0$, cf.\ Example \ref{ex.wass2} b).}
\label{fig:exponential_wass2}
\end{figure}

\item[c)] Let $p\in(0,1)$ and $Y_i\sim{\rm B}(1,p)$ for $i=1,\ldots, n$.\ Then, $Y_i^2=Y_i$ for $i=1,\ldots, n$, so that
\[
 \cW_2(\P_X,\P_Y)^2\leq 2(1-\ga)\big(1-\E[Y_i]\big)\TVaR^\ga(Y_i)=\begin{cases}
 2p(1-p),&\ga< 1-p,\\
 2(1-\ga)(1-p),& \ga\geq 1-p.\end{cases}
\]
In the case where $1-p\leq \ga$, we thus obtain
\[
\cW_2(\P_X,\P_Y)\leq \ep\quad \text{if}\quad \ga\geq  1-\frac{\ep^2}{2(1-p)}.
\]
\end{enumerate}
\end{example}

\subsection{Wasserstein Distance for Exchangeable Bernoulli Random Variables}
 In this subsection, we aim to identify appropriate values of the radius $\ep>0$ for the Wasserstein constraints derived in the previous subsection by analyzing the distance between exchangeable random vectors in the setup from Section \ref{sec:credit-risk} below.\ To that end, let $p\in [1,\infty)$ and $\R^n$ be endowed with the norm $|x|:=\big(\frac1n \sum_{i=1}^n|x_i|^p\big)^{\frac1p}$ for $x\in \R^n$.\
In this section, we consider two $\{0,1\}^n$-valued random vectors $X=(X_1,\ldots,X_n)$ and $Y=(Y_1,\ldots,Y_n)$.\ We assume that the joint laws $\P_X$ and $\P_Y$ are exchangeable, and define 
$$
S_X:=\sum_{i=1}^n X_i,
\quad\text{and}\quad
S_Y:=\sum_{i=1}^n Y_i.
$$
Recall that the laws $\P_X$ and $\P_Y$ are exchangeable if and only if, for all $k,l\in \{0,\ldots, n\}$ and $x,y\in \{0,1\}^n$ with $k=\sum_{i=1}^nx_i$ and $l=\sum_{i=1}^ny_i$,
\begin{equation}\label{eq.dist.exchangeable}
\P(X=x)=\frac{\P(S_X=k)}{\binom{n}{k}}\quad \text{and}\quad \P(Y=y)=\frac{\P(S_Y=l)}{\binom{n}{l}}.
\end{equation}
Since the random vectors $X$ and $Y$ take values in $\{0,1\}^n$,
\[
|X-Y|^p=\frac1n\sum_{i=1}^n |X_i-Y_i|\geq \frac1n |S_X-S_Y|,
\]
so that
\[
\cW_p(\P_X,\P_Y)\geq \bigg(\frac1n\sum_{k=0}^{n-1} \big|F_{S_X}(k)-F_{S_Y}(k)\big|\bigg)^{\frac1p}.
\]
On the other hand, given a co-monotonic realization of $S_X$ and $S_Y$, for $k,l\in \{0,\ldots, n\}$ and $x,y\in \{0,1\}^n$ with $k=\sum_{i=1}^nx_i$ and $l=\sum_{i=1}^ny_i$, define
\[
\P\big(\{X=x\}\cap \{Y=y\}\big):=\frac{\P\big(\{S_X=k\}\cap \{S_Y=l\}\big)}{\binom{n}{k}\binom{k}{l}}
\]
if $k\geq l$ and $x-y\geq 0$, and
\[
\P\big(\{X=x\}\cap \{Y=y\}\big):=\frac{\P\big(\{S_X=k\}\cap \{S_Y=l\}\big)}{\binom{n}{l}\binom{l}{k}}
\]
if $l\geq k$ and $y-x\geq 0$.\ Since  $\binom{n}{k}\binom{k}{l}=\binom{n-l}{k-l}\binom{n}{l}$ if $k\geq l$ and $\binom{n}{l}\binom{l}{k}=\binom{n-k}{l-k}\binom{n}{k}$ if $l\geq k$,
it follows that
\[
\P(X=x)=\frac{\P(S_X=k)}{\binom{n}{k}}\quad \text{and}\quad \P(Y=y)=\frac{\P(S_Y=l)}{\binom{n}{l}},
\]
so that, by \eqref{eq.dist.exchangeable}, $X\sim \P_X$ and $Y\sim \P_Y$.\ By construction,
\[
|X-Y|^p=\frac1n|S_X-S_Y|,
\]
and we therefore obtain
\begin{equation}\label{eq.wass.exchangeable}
\cW_p(\P_X,\P_Y)= \bigg(\frac1n\sum_{k=0}^{n-1} \big|F_{S_X}(k)-F_{S_Y}(k)\big|\bigg)^{\frac1p}.
\end{equation}

\begin{remark}\label{rem.wass2}
For the case of the $2$-Wasserstein distance, if we consider the law $\P_X$ as a plausible alternative for $\P_Y$, the right-hand side in \eqref{eq.wass.exchangeable} yields a reasonable choice for the radius $\ep>0$ in Theorem \ref{th: constrainedwc}. The basic reasoning is that since $\P_X$ is considered plausible, all other models with smaller or equal distance to the reference model are plausible as well; see also \cite{BernardPesentiVanduffel2024}. \ By Example \ref{ex.wass2} c), this leads to the explicit condition
\begin{equation}\label{eq.wass.cond.exchangeable}
1-\ga\leq \frac{\sum_{k=0}^{n-1} \big|F_{S_X}(k)-F_{S_Y}(k)\big|}{2n\big(1-\P(Y_1=1)\big)}
\end{equation}
for $\ga\in (0,1)$.
\end{remark}

\section{Model Risk in Credit Risk Portfolios}\label{sec:credit-risk}

\subsection{Set-up}\label{sec:setup-below}

We consider a portfolio of $n\in \N$ borrowers, and denote by $X_i$ the indicator function of the default event for borrower $i=1,\ldots, n$ over a one-year time horizon.\ The total loss caused by a default coincides
with the remaining exposure (net of recoveries) denoted by $\text{Exposure}_i$. A precise computation of a tail risk measure of the portfolio loss $\sum_{i=1}^n \text{Exposure}_i \cdot X_i$
can only be obtained if one knows the joint distribution of the default vector
\[
(X_1,X_2,\ldots,X_n).
\]

While financial institutions typically use models that allow for the specification of default probabilities $p_i$ and thus of the distribution $F_i$ of the variables $X_i$, their joint distribution remains difficult to specify. The lack of sufficiently rich default statistics usually makes it impossible to estimate joint default probabilities with reasonable accuracy. In other words, all models that assess the risk of credit portfolios necessarily require additional ad hoc and hard-to-justify assumptions to describe the full dependence structure. In this regard, many financial institutions, as well as the Basel~III and Solvency~II regulatory frameworks, rely on ``Merton's model of the firm'' when assessing the default risk of a portfolio of loans. The basic idea is very simple. A default is an event in which the net asset value becomes too low or, equivalently, the net liability becomes too high, i.e., default of the $i$-th risk occurs when
\[
L_i > \tau_i,
\]
where $L_i$ is a normalized random net liability and $\tau_i$ is the threshold value such that
\[
p_i = \mathbb{P}(L_i > \tau_i)
\]
and thus $X_i=\eins_{\{ L_i> \tau_i\}}$.
We refer to, e.g., \cite{FreyMcNeil2003} for more details. Within this framework, a key issue in risk assessment is the estimation of the joint distribution of the borrowers' net liabilities. In this respect, the standard specification of Merton's model consists in assuming that, after normalization, the random variables $L_i$ are $\mathcal{N}(0,1)$ and exhibit Gaussian dependence fully specified by a correlation matrix. Under this assumption, one also has that
$(X_1,\dots,X_n)\sim c(F_{X_1},\ldots, F_{X_n})$, i.e., within structural models, the default indicators inherit the Gaussian dependence structure from the net liabilities.\  

Then, given a risk measure $R$, the portfolio risk assessment is obtained by computing 
the following quantity:
\[
R\left( \sum_{i=1}^n \text{Exposure}_i \cdot X_i\right),
\]
To study the impact of dependence uncertainty, we consider as admissible all probability measures \(\P_X\) with prescribed marginal distributions \(F_{X_1},\ldots,F_{X_n}\) such that
\[
\cW_2\big(\P_X, c(F_{X_1},\ldots,F_{X_n})\big)\leq \ep, 
\]
where we take the reference copula \(c\) to be Gaussian. 

The radius \(\ep\) is then calibrated as the Wasserstein distance between the Gaussian reference model and an alternative model in which \(\P_X\) is driven by a Student's \(t\)-copula with \(\nu\) degrees of freedom. In other words, if such a \(t\)-copula model for \(\P_X\) is deemed plausible, then all copula models whose induced distribution for $(X_1,\dots,X_n)$ lies at an equal or smaller Wasserstein distance from the Gaussian benchmark \(c(F_{X_1},\ldots,F_{X_n})\) are also regarded as admissible; see also \cite{BernardPesentiVanduffel2024}.

\subsection{Assessing VaR and TVaR}\label{sec.TVAR}

VaR and TVaR are the most prominent risk measures used in operational risk management and are also central to regulatory capital frameworks. For instance, the Swiss Solvency Test (SST) adopts TVaR at the \(99\%\) confidence level as its solvency risk measure \cite{FINMA2006SST}, while the Basel framework for credit risk capital relies on VaR at the \(99.9\%\) confidence level, cf.\ \cite{BCBS2019CRE31}. In line with this, hereafter we consider $\alpha=0.999$ for $\VaR^\alpha$ and $\alpha=0.99$ for $\TVaR^\alpha$.

We first focus on a portfolio of $n=500$ loans.\ Furthermore, a Gaussian copula with a homogeneous correlation matrix, in which all pairwise correlations are equal to $\rho=0.1$ is used to model the dependence among the liabilities $L_i$.\ The critical default threshold is $\tau_i=F^{-1}_{L_i}(0.99)$, which yields a marginal default probability $\P(X_i=1)=0.01$ for every loan $i=1,\ldots, n$ in the portfolio.\ For simplicity, all exposures are set to be equal to one. Theorem \ref{th: boundeddistance} guarantees that, for every $\ep>0$, there exists $\gamma\in(0,1)$ such that all couplings with $(c,\ga)$-hidden dependence do not deviate more than $\ep$ from the reference model.\ Thus, also the coupling with  $(c,\ga)$-hidden dependence and co-monotonic upper tails is an admissible model for a suitable choice of $\gamma$.\ In this regard, we point out that co-monotonic upper tails are not unrealistic in that there is empirical evidence suggesting the presence of extreme (tail) co-movements among risks, cf. \cite{ ContKokholm2013, ContWagalath2016,DasUppal2004}.

Table \ref{tab:percent_table} reports, for \(\text{VaR}^{0.999}\) and \(\text{TVaR}^{0.99}\), the ratio between the portfolio risk measure computed under a \((c,\gamma)\)-hidden dependence structure and co-monotonic upper tails satisfying the \(\cW_2\) constraint and the corresponding value computed under the reference Gaussian copula \(c\).\ The radius is computed using the exact formula \eqref{eq.wass.exchangeable}.\ This reflects the fact that all models whose distance from the reference model does not exceed the distance of the model for \((X_1,\ldots,X_n)\), obtained by replacing only the Gaussian copula with a \(t\)-copula with \(\nu\) degrees of freedom, are considered plausible.\ In this regard, we note that the larger the number of degrees of freedom, the more closely the \(t\)-copula resembles the Gaussian copula, and hence the smaller the resulting radius.\ If empirical data are available and a copula model must be selected from the data, it may be difficult to statistically distinguish a \(t\)-copula with \(\nu=20\) degrees of freedom from a Gaussian copula.\ For such a high number of degrees of freedom, the \(t\)-copula is very close to the Gaussian copula, especially when the correlation parameter \(\rho\) is small.\ To determine the value of \(\gamma\) associated with a given radius \(\ep\), we use the relation between \(\ep\) and \(\gamma\) derived in Remark \ref{rem.wass2}(c). 

\begin{table}[h]
\centering
\renewcommand{\arraystretch}{1.5}
\captionsetup{skip=15pt}
\begin{tabular}{|c|c|c|c|c|c|}
\hline
$\nu$, t-copula  & $\cW_2$ radius $\ep$ &  Compatible $\gamma$ & Underest.\ $\TVaR^{0.99}$ & Underest.\ $\VaR^{0.999}$ \\
\hline
3 & 0.1063 & 0.9944 & 883\% & 1219\% \\
5 & 0.0917 & 0.9958 & 740\% & 1219\% \\
10 & 0.0733 & 0.9973 & 534\% & 1219\% \\
20 & 0.0550 & 0.9985 & 362\% & 1219\% \\
\hline
\end{tabular}

\caption{Homogeneous portfolio of credit loans: size $n=500$, default probability $p=0.01$ and asset correlation $\rho=0.10$.\ For each 2-Wasserstein radius $\ep$, the table reports a compatible value of $\gamma$ such that the $(c,\ga)$-hidden dependence model with co-monotonic tail copula $c^{\rm tail}$ is feasible.\ Columns 4 and 5 give the corresponding underestimation ratios relative to the reference Gaussian copula $c$ for $\operatorname{TVaR}_{0.99}(\sum_{i=1}^n X_i)$ and $\operatorname{VaR}_{0.999}(\sum_{i=1}^n X_i)$, respectively.\ Both ratios are estimated via Monte Carlo simulation.}
\label{tab:percent_table}
\end{table}

The results in Table \ref{tab:percent_table} show that even relatively small values of $\ep>0$, corresponding to a small perturbation of the dependence structure, can translate into a significant underestimation of the TVaR$^{0.99}$  and VaR$^{0.999}$ of the portfolio loss.\ Notice that since $\gamma > 0.99$, the (c,$\gamma$)-hidden model with a co-monotonic tail copula is merely one of several plausible models, but not necessarily the worst-case model for TVaR$^{0.99}$.\ Nevertheless, the extent of possible underestimation is substantial, even when a $t_{20}$-copula is used.

The results for VaR are particularly striking.\ In all cases, the degree of underestimation is more than tenfold.\ In fact, as $\gamma < 0.999$ holds, this value corresponds to the unconstrained upper bound for $\VaR^{0.999}$, which is easily seen to be equal to 500 and this bound is also attained by the (c,$\gamma$)-model with co-monotonic upper tails, whereas the VaR under the reference model is only 41.

These findings show that even very small deviations from a reference model can substantially inflate risk measures when very high confidence levels are used. In other words, we provide evidence that capital requirements based on extreme tail-risk measures, as commonly used in practice, are highly sensitive to even minor model modifications. Accordingly, we recommend caution when using tail-risk measures that rely on high confidence levels.

To shed further light on this issue, we assess, for large homogeneous portfolios, the extent to which VaR may be underestimated for different values of the default probability \(p\) and the asset correlation $\rho$. To this end, we use the asymptotic formula of \cite{Vasicek2002LoanPortfolioValue} for the VaR of a large homogeneous credit portfolio under the Merton model, expressed as a percentage of total exposure:

\begin{equation}\label{eq.asymptotic.vasicek} {\lim_{n\to \infty}}\VaR^{0.999}\bigg({\frac1n}\sum_{i=1}^n X_i\bigg)
=
\Phi\bigg(
\frac{\Phi^{-1}(p)+\sqrt\rho\,  \Phi^{-1}(0.999)}
{\sqrt{1-\rho}}
\bigg).\end{equation}
We point out that the latter also enters the formula for the computation of risk-weighted assets (RWAs) with $\rho=R(p)$, described in Article 153 of the Capital Requirements Regulation (CRR) \cite{CRR}.\ If the compatible $\ga$ is less than $0.999$, \eqref{eq.asymptotic.vasicek} provides a theoretical upper bound for the $\VaR^{0.999}$-underestimation ratio for a large number of obligors:
\begin{equation}\label{eq.var.underestimation}
\lim_{n\to \infty} \frac{\VaR^{0.999}(\sum_{i=1}^n Z_i)}{\VaR^{0.999}(\sum_{i=1}^n X_i)}=\Bigg(\Phi\bigg(
\frac{\Phi^{-1}(p)+\sqrt \rho\,  \Phi^{-1}(0.999)}
{\sqrt{1-\rho}}
\bigg)\Bigg)^{-1},
\end{equation}
where, for each $n\in \N$, $(Z_1,\ldots, Z_n)$ is a random vector with $(c,\ga)$-hidden dependence and co-monotonic tail copula $c^{\rm tail}$, representing a joint distribution that yields the worst-case aggregate risk for $\VaR^{0.999}$.\ In this context, we also refer to \cite[Theorem 4.2]{nendel2026asymptotic} for an analogous formula for arbitrary $\ga$-tail risk measures in the case $\rho=0$.

 The following table reports the \(\cW_2\) radius $\ep$, the corresponding compatible hidden-dependence level $\ga$, and the implied VaR underestimation
ratio for an exchangeable Bernoulli portfolio with $n=50{,}000$ obligors and a
$t$-copula alternative with $\nu =20$ degrees of freedom. The last column reports the underestimation relative to the worst-case VaR {according to \eqref{eq.var.underestimation}}.\ Note that, as $\ga<0.999$, we are guaranteed that the worst-case model corresponds to the unconstrained upper bound.

\begin{table}[htbp]
\centering

\label{tab:epsilon-gamma-underest-var-n50000}
\begin{tabular}{cccccc}
\hline
$p$ & $\rho$ & $\nu$ & $\varepsilon$ & Compatible $\gamma$ & Underest.\ $\mathrm{VaR}^{0.999}$ \\
\hline
0.0075 & 0.05 & 20 & 0.059 & 0.9982 & 2703\% 
\\
0.0075 & 0.10 & 20 & 0.053 & 0.9986 & 1599\% \\
0.0075 & 0.15 & 20 & 0.049 & 0.9988 & 1110\% \\
0.0100 & 0.05 & 20 & 0.065 & 0.9979 & 2142\% \\
0.0100 & 0.10 & 20 & 0.058 & 0.9983 & 1290\% \\
0.0100 & 0.15 & 20 & 0.053 & 0.9986 & 907\% \\
0.0125 & 0.05 & 20 & 0.070 & 0.9975 & 1791\% \\
0.0125 & 0.10 & 20 & 0.062 & 0.9981 & 1095\% \\
0.0125 & 0.15 & 20 & 0.057 & 0.9984 & 778\% \\
\hline
\end{tabular}
\caption{Homogeneous portfolio of $n=50,000$ credit loans. For different values of the default probability $p$ and the asset correlation $\rho$, the last column of the table reports the corresponding underestimation ratio for small perturbations (driven by a $t_{20}$-copula) relative to the reference Gaussian copula $c$ for $\operatorname{VaR}_{0.999}(\sum_{i=1}^n X_i)$.\ This ratio is computed via the asymptotic formula \eqref{eq.var.underestimation}.}
\end{table}

% \appendix

%\bibliographystyle{abbrv}
%\bibliography{lit}

\end{document}